\documentclass[11pt]{article}
\usepackage{amsmath,amsthm}
\usepackage{amssymb}
\usepackage[dvipdfmx]{graphicx,color}
\allowdisplaybreaks
\nopagebreak[3]
\newcommand{\newsection}[1]
{\section{#1}\setcounter{theorem}{0} \setcounter{equation}{0} \par\noindent}

\newtheorem{theorem}{Theorem}

\newtheorem{lemma}[theorem]{Lemma}
\newtheorem{corollary}[theorem]{Corollary}

\newtheorem{remark}[theorem]{Remark}
\newtheorem{claim}[theorem]{Claim}
\newcommand{\beq}{ \begin{equation} }
\newcommand{\eeq}{ \end{equation} }

\newcommand{\hdot}{{\dot H}}
\newcommand{\xdot}{{\dot X}}
\newcommand{\br}{{\mathbb R}}
\newcommand{\bc}{{\mathbb C}}

\newcommand{\brn}{ { \mathbb{R}^n } }
\renewcommand{\Re}{\operatorname{Re}}
\renewcommand{\Re}{\operatorname{Re}}

\newcommand{\two}{I\hspace{-1mm}I}
\newcommand{\three}{I\hspace{-1mm}I\hspace{-1mm}I}

\title{
The Cauchy problem for semi-linear Klein-Gordon 
\\ 
equations in Friedmann-Lema\^itre-Robertson-Walker spacetimes
}
\author
{
Makoto NAKAMURA
\thanks
{Graduate School of Information Science and Technology,  
Osaka University, 
1-5 Yamadaoka, Suita, Osaka 565-0871, JAPAN.  
E-mail:  \texttt{makoto.nakamura.ist@osaka-u.ac.jp} }
\ \ and \ 
Takuma YOSHIZUMI
\thanks
{Graduate School of Information Science and Technology,  
Osaka University, 
1-5 Yamadaoka, Suita, Osaka 565-0871, JAPAN.  
E-mail:  \texttt{yoshizumi.takuma@ist.osaka-u.ac.jp} }
}

\begin{document}

\maketitle

\begin{abstract}
The Cauchy problem for semi-linear Klein-Gordon equations is considered in Friedmann-Lema\^itre-Robertson-Walker spacetimes. 
The local and global well-posedness of the Cauchy problem is considered in Sobolev spaces.
The non-existence of global solutions is also considered.
\end{abstract}

\noindent
{\it Mathematics Subject Classification (2020)}: 
Primary 35L05; Secondary 35L71, 35Q75. 

\vspace{5pt}

\noindent
{\it Keywords} : 
semilinear Klein-Gordon equation, Cauchy problem, 
Friedmann-Lema\^itre-Robertson-Walker spacetime

%

\section{Introduction}

We consider the Cauchy problem of semi-linear Klein-Gordon equations in Friedmann-Lema\^itre-Robertson-Walker spacetimes (FLRW spacetimes for short).
FLRW spacetimes are solutions of the Einstein equations 
with the cosmological constant under the cosmological principle.
They describe the spatial expansion or contraction, 
and yield some important models of the universe.
Let $n\ge1$ be the spatial dimension, 
$a(\cdot)>0$ be a scale-function defined on an interval $[0,T_0)$ for some $0<T_0\le\infty$,
$c>0$ be the speed of light.
We consider the metrics $\{g_{\alpha\beta}\}$ of FLRW spacetimes expressed by 
\begin{equation}
\label{Intro-RW}
-c^2(d\tau)^2
=
\sum_{0\le \alpha,\beta\le n}g_{\alpha\beta}dx^\alpha dx^\beta
:=
-c^2(dt)^2+
a^2(t)
\sum_{j=1}^n (dx^j)^2,
\end{equation}
where we have put the spatial curvature as zero,  
the variable $\tau$ denotes the proper time, 
$x^0=t$ is the time-variable 
(see, e.g., 
\cite{Carroll-2004-Addison, DInverno-1992-Oxford}).
When $a$ is a positive constant, 
the spacetime with \eqref{Intro-RW} reduces to the Minkowski spacetime.


We recall a derivation of the Klein-Gordon equation.
Let $g :={\rm det}(g_{\alpha\beta})$ be the determinant of matrix $(g_{\alpha\beta})$, 
and let $(g^{\alpha\beta})$ be the inverse matrix of  
$(g_{\alpha\beta})$.
For a real number $1<p<\infty$ and $\lambda \in \br$, 
let $V$ be the real valued potential function given by 
\beq
\label{Def-V}
V(z):=\frac{2\lambda}{p+1} |z|^{p+1}
\eeq
for $z\in \mathbb{C}$.
Then, the equation of motion of massive scalar field described by a function $v$ 
with mass $m\ge 0$, the potential $V$ 
and the metric 
$\{g_{\alpha\beta}\}$ is given by the form 
$(\sqrt{|g|})^{-1} \partial_\alpha (\sqrt{|g|} g^{\alpha\beta}\partial_\beta v)=m^2v+V'(v)$, namely, 
\beq
\label{Cauchy-0}
\frac{1}{c^2}\partial_t^2v+\frac{n\dot{a}}{c^2a}\partial_tv-\frac{1}{a^2} \Delta v +m^2v+V'(v)=0,
\eeq
where 
$\Delta:=\sum_{j=1}^n \partial_j^2$ denotes the Laplacian, 
and 
$V'(v):=\lambda |v|^{p-1}v$.
We denote the first and second derivatives of one variable function $a$ by $\dot{a}$ and $\ddot{a}$, respectively.
By the transformation $u:=va^{n/2}$, 
\eqref{Cauchy-0} is rewritten as 
\beq
\label{Eq-KGLambda}
c^{-2}\partial_t^2u-a^{-2}\Delta u
+M^2u
+a^{n/2}V'(a^{-n/2}u)=0,
\eeq
where the function $M$ is called the curved mass defined by  
\beq
\label{Def-M}
M^2
=
m^2-\frac{n(n-2)}{4c^2}\left( \frac{\dot{a}}{a} \right)^2
-\frac{n}{2c^2}\cdot\frac{\ddot{a}}{a},
\ \ \ \ 
M:=\sqrt {M^2}.
\eeq
We note that $M$ is purely imaginary number when $M^2<0$.

Generalizing the function $V'$ in \eqref{Cauchy-0} into the power-type function $f$ given by 
\beq
\label{Def-f}
\lambda \in \bc,
\ \ f(z):=\lambda |z|^{p-1}z
\ \ 
\mbox{or}
\ \ 
f(z):=\lambda |z|^{p}
\eeq
for $z\in \mathbb{C}$,
we obtain the semilinear Klein-Gordon equation 
\beq
\label{Eq-KGV}
c^{-2}\partial_t^2u-a^{-2} \Delta u+M^2u+a^{n/2}f(a^{-n/2}u)=0.
\eeq

We consider the Cauchy problem of \eqref{Eq-KGV} given by   
\beq
\label{Cauchy}
\left\{
\begin{array}{l}
(c^{-2}\partial_t^2-a^{-2}(t) \Delta+M^2(t))u(t,x)+a^{n/2}(t)f(a^{-n/2}(t)u(t,x))=0,
\\
u(0,\cdot)=u_0(\cdot),\ \ \partial_tu(0,\cdot)=u_1(\cdot)
\end{array}
\right.
\eeq
for $(t,x)\in [0,T)\times\br^n$, 
where 
$u_0$ and $u_1$ are given initial data.

Here, we expect some dissipative effects when $\dot{a}>0$ in \eqref{Cauchy-0} 
by the dissipative term $n\dot{a}\partial_t v/c^2a$.
Namely, the spatial expansion yields a dissipative effect to the Klein-Gordon equation.
Actually, we have the property of dissipative wave equations in the energy estimate \eqref{Ineq-EnergyEs-Formal}, below, 
and we use it to solve the Cauchy problem.

\vspace{9pt}


Now, we introduce some concrete examples of the scale-function $a$ and the curved mass $M$.
For $\sigma\in\mathbb{R}$ and the Hubble constant $H\in \br$, 
we put
\begin{equation}
\label{R-Def-T_0}
T_0:= 
\begin{cases}
\infty & \mbox{if}\ \ (1+\sigma)H\ge0,
\\
-\frac{2}{n(1+\sigma)H} & \mbox{if}\ \ (1+\sigma)H<0,
\end{cases}
\end{equation}
and define $a(\cdot)$ by 
\begin{equation}
\label{Def-a}
a(t):=
\begin{cases}
a_0\left\{1+\frac{n(1+\sigma)Ht}{2}\right\}^{2/n(1+\sigma)} & \mbox{if}\ \ \sigma\ne-1,
\\
a_0\exp(Ht) & \mbox{if}\ \ \sigma=-1
\end{cases}
\end{equation}
for $0\le t<T_0$. 
We note $a_0=a(0)$ and $H=\dot{a}(0)/a(0)$.
This scale-function $a(\cdot)$ describes the Minkowski spacetime when $H=0$ 
(namely, $a(\cdot)$ is a constant $a_0$),
the expanding space when $H>0$ with $\sigma\ge -1$,
the blowing-up space when $H>0$ with $\sigma< -1$ 
(the ``Big-Rip'' in cosmology), 
the contracting space when $H<0$ with $\sigma\le -1$, 
and the vanishing space when $H<0$ with $\sigma> -1$ 
(the ``Big-Crunch'' in cosmology).
It describes the de Sitter spacetime when $\sigma=-1$ 
(see, e.g., \cite{Nakamura-2020-OsakaJMath}).
The corresponding curved mass $M$ defined by \eqref{Def-M} is rewritten as  
\beq
\label{M-FLRW}
M^2
=
m^2+\sigma\left(\frac{nH}{2c}\right)^2\cdot 
\left\{1+\frac{n(1+\sigma)Ht}{2}\right\}^{-2}
\eeq
(see (1) in Lemma \ref{Lem-11}, below).
For $H\in \br$, $\sigma\in \br$ and $m\ge0$, 
put 
\beq
\label{Def-T1}
T_1:=
\begin{cases}
\infty & \mbox{if}\ \ (1+\sigma)H\ge 0,
\\
-\frac{2}{n(1+\sigma)H}\left( 1-\frac{ \sqrt{|\sigma|}n|H| }{2cm} \right) 
& \mbox{if}\ \ (1+\sigma)H<0, \ \sigma<0, \ m>\frac{\sqrt{|\sigma|} n |H|}{2c},
\\
-\frac{2}{n(1+\sigma)H} & \mbox{otherwise},
\end{cases}
\eeq
which satisfies $0<T_1\le T_0$.
We say that the solution $u$ of \eqref{Cauchy} is a global solution if $u$ exists on the time-interval $[0,T_0)$ 
since $T_0$ is the end of the spacetime.
We note that the squared curved mass $M^2$ changes its sign from plus to minus at $T_1$ 
when $(1+\sigma)H<0$, $\sigma<0$, $m>\sqrt{|\sigma|}n|H|/2c$ by \eqref{M-FLRW}.

\vspace{10pt}

We refer to closely related results on the Cauchy problem \eqref{Cauchy}.
In the expanding de Sitter spacetime ($H>0$ and $\sigma=-1$), Yagdjian \cite{Yagdjian-2012-JMAA} has shown 
small global solutions for \eqref{Cauchy}, 
provided that the norm of initial data 
$\|u_0\|_{H^s(\brn)}+\|u_1\|_{H^s(\brn)}$ is sufficiently small for given $s>n/2\ge 1$ 
(see \cite{Yagdjian-2013-Springer} for the system of the equations, 
and \cite{Galstian-Yagdjian-2017-NA-RWA} in the case of the Riemann metric space for each time slices).
In \cite{Nakamura-2014-JMAA}, the energy solutions 
for $u_0\in H^1(\brn)$ and $u_1\in L^2(\brn)$ 
have been shown, 
and the effect of the spatial expansion was characterized as the dissipative property in energy estimates.
This result was extended to the case of general FLRW spacetimes  
in \cite{Galstian-Yagdjian-2015-NA-TMA}.
Our main aim is to extend this result in general frame of Sobolev spaces with any nonnegative regularity.
We give some detailed results to show several effects of spatial variance on the problem.

In the case of  the asymptotic de Sitter spacetime, 
Baskin \cite{Baskin-2010-ProcCMAANU, Baskin-2012-AHP} has shown global solutions with small energy, 
which was further investigated on the semilinear term including derivatives of solution 
by Hintz and Vasy in 
\cite{Hintz-Vasy-2015-AnalysisPDE}.
We refer to 
\cite{Tsuchiya-Nakamura-2019-JCompApplMath} 
for a numerical simulation for the semilinear Klein-Gordon equation, 
and 
\cite{Nakamura-2015-JDE,Nakamura-2021-JDE} 
on the Cauchy problem  
for 
semilinear Schr\"odinger equations 
and 
semilinear Proca equations  
in the de Sitter spacetime.

\vspace{10pt}

We denote the Lebesgue space by $L^q(I)$ for an interval $I\subset \br$ 
and $1\le q\le \infty$ with the norm 
\[
\|Y\|_{L^q(I)}:=
\begin{cases}
\left\{
\int_I |Y(t)|^q dt
\right\}^{1/q} & \mbox{if}\ \ 1\le q<\infty,
\\
{\mbox{ess.}\sup}_{t\in I} |Y(t)| 
& \mbox{if}\ \ q=\infty.
\end{cases}
\]
We use the Sobolev space $H^\mu(\brn)$, the homogeneous Sobolev space $\dot{H}^\mu(\brn)$, 
and the homogeneous Besov space 
$\dot{B}^\mu_{r,s}(\brn)$ of order $\mu\ge0$ for $1\le r, s\le \infty$ 
(see \cite{Bergh-Lofstrom-1976-Springer} for their definitions).
Put $\nabla:=(\partial_1,\cdots,\partial_n)$.

Firstly, we show the well-posedness of the Cauchy problem \eqref{Cauchy}.
For $\mu\ge0$ and $T>0$, 
we define a function space $X(T)$ by
\begin{equation}
\label{X-Set}
X(T):=\{u:\, \|u\|_{\xdot^\nu(T)}<\infty \,\ \mbox{for} \,\ \nu=0,\mu_0,\mu \}
\end{equation}
with the metric $d(u,v):=\|u-v\|_{\xdot^0(T)}$ for $u,v\in X(T)$,
where 
\begin{multline}
\label{Def-XT}
\|u\|_{\xdot^\nu(T)}
:=
\max\big\{
c^{-1}\|\partial_tu\|_{L^\infty((0,T),\hdot^\nu(\brn))}, 
\|a^{-1} \nabla u\|_{L^\infty((0,T),\hdot^\nu(\brn))}, 
\\
\|Mu\|_{L^\infty((0,T),\hdot^\nu(\brn))}, 
\|\sqrt{\dot{a}a^{-3}} \nabla u\|_{L^2((0,T),\hdot^\nu(\br^n))},
\\
\|\sqrt{-\dot{M}M} u\|_{L^2((0,T),\hdot^\nu(\br^n))}
\big\}
\end{multline}
for $\nu=0,\mu_0,\mu$, 
provided $\dot{a}\ge0$, $M^2\ge0$ and $\dot{M}\le 0$.

We show our results for a general scale-function $a$ and a curved mass $M$ 
as theorems, 
and then we show our results 
for the examples of \eqref{Def-a} and \eqref{M-FLRW} as their corollaries.

We consider the existence of local and global solutions.
The existence of small energy solutions for $\mu_0=\mu=0$ with $1+4/n\le p\le 1+2/(n-2)$ has been considered in 
\cite[Theorem 1.1]{Galstian-Yagdjian-2015-NA-TMA} 
with $V'$ in \eqref{Eq-KGLambda} replaced by $\Gamma V'$ for some weighted function $\Gamma=\Gamma(t)$,
while our result does not need this weight function.
The following result generalizes it into the high regularity $0\le \mu_0\le \mu<\infty$.


\begin{theorem}[Local and global solutions]
\label{Thm-13-Cor-14}
Let $n\ge1$.
Let $0\le \mu_0<n/2$ if $n=1,2$, and let $0\le \mu_0<n/2-1$ if $n\ge3$.
Let $\mu_0\le \mu<\infty$.
Let $p$ satisfy 
\beq
\label{Condition-p}
1\le p
\begin{cases}
<\infty & \mbox{if}\ \  n=1,2,\\
\le p(\mu_0):=1+\frac{2}{n-2\mu_0-2} & \mbox{if}\ \ n\ge 3.
\end{cases}
\eeq
Let $\lambda\in \mathbb{C}$, and let $f$ satisfy \eqref{Def-f}.
Assume $\mu<p$ unless 
$f(u)=\lambda |u|^{p-1}u$ for odd $p$, 
or 
$f(u)=\lambda |u|^p$ for even $p$.
Let $T_0$, $T_1$, $a$ and $M$ satisfy 
\beq
\label{Condition-aM}
\begin{cases}
0<T_1\le T_0\le \infty,
\\
a\in C^2([0,T_0),(0,\infty)),
\ \ 
\dot{a}\ge0, 
\\ 
M\in C^1([0,T_1),(0,\infty)), 
\ \ 
\dot{M}\le0.
\end{cases}
\eeq
Put $a_0:=a(0)$ and $M_0:=M(0)$.
Then the following results hold.

(1) 
(Local solutions.) 
For any $u_0\in H^{\mu+1}(\brn)$ and $u_1\in H^{\mu}(\brn)$, 
there exists a unique solution $u$ of \eqref{Cauchy} in 
$C([0,T),H^{\mu+1}(\brn))\cap C^1([0,T),H^\mu(\brn))$ $\cap X(T)$ 
for any $T$ with 
$0<T\le T_1$ and 
\beq
\label{Condition-T}
CcA(T)(C_0D_{\mu_0})^{p-1}\le 1,
\eeq 
where 
$C_0$ and $C$ are positive constants independent of $u_0$, $u_1$ and $T$,
\beq
\label{Def-A}
A(T)
:=
M(T)^{-\delta} 
\left\|
a^{-\mu_0(p-1)}
\left(\frac{\dot{a}}{a}\right)^{1/q_\ast-1}
\right\|_{L^{q_\ast}(0,T)},
\eeq
\beq
\label{Def-delta-qAst}
\delta:=1-\frac{(p-1)(n-2\mu_0-2)}{2},
\ \  
\frac{1}{q_\ast}:=1-\frac{(p-1)(n-2\mu_0)}{2q},
\eeq
\beq
D_{\mu_0}:=c^{-1}\|u_1\|_{\dot{H}^{\mu_0}(\brn)}
+ca_0^{-1}\|\nabla u_0\|_{\dot{H}^{\mu_0}(\brn)}
+M_0\|u_0\|_{\dot{H}^{\mu_0}(\brn)},
\eeq
and $q$ is an arbitrarily fixed number with 
\beq
\label{Def-q2}
0\le \frac{1}{q}\le \min\left\{\frac{1}{2},\frac{2}{(p-1)(n-2\mu_0)}\right\}.
\eeq

(2) 
(Continuous dependence on data.)
Let $v$ be the solution of the Cauchy problem \eqref{Cauchy} obtained in (1) for initial data $v_0:=v(0,\cdot)$ and $v_1=\partial_tv(0,\cdot)$.
If $(v_0,v_1)$ converges to $(u_0,u_1)$ in $H^{\mu+1}(\brn)\oplus H^\mu(\brn)$, 
then there exists $T>0$ such that $\|u-v\|_{\dot{X}^0(T)}$ tends to zero.

(3) 
(Small global solutions.)
If 
\beq
\label{Condition-SmallGlobal}
T_1=T_0, 
\ \ 
\inf_{0<t<T_0} M(t)>0,
\ \ 
Cc (C_0D_{\mu_0})^{p-1} \sup_{0<t<T_0} A(t) 
\le 1,
\eeq
then the solution $u$ obtained in (1) is a global solution, namely, $T$ can be taken as $T=T_0$, 
and there exists a free solution $u_+$ with $(c^{-2}\partial_t^2-a^{-2}\Delta +M^2)u_+=0$ 
such that 
\beq
\label{Thm-13-1000}
\max_{\substack{0\le \theta\le 1 \\ k=0,1}}
\left(\frac{M(t)}{a(t)}\right)^{\theta}\cdot 
\|\partial_t^k\big(u-u_+\big)(t)\|_{H^{\mu-1+\theta}(\brn)}
\eeq
converges to zero as $t$ tends to $T_0$.

(4) 
(Large global solutions.)
If $T_1=T_0$, $\mu_0=0$, $\lambda\ge0$ and $f(u)=\lambda |u|^{p-1}u$, 
then the solution $u$ obtained in (1) is a global solution, namely, $T=T_0$.
\end{theorem}


The following Corollary \ref{Thm-16-Cor-17} is derived from Theorem \ref{Thm-13-Cor-14} 
for the scale-function $a$ defined by \eqref{Def-a}.
It generalizes the case $\sigma=-1$ with $H>0$ and $\mu_0=\mu=0$, namely, 
the energy solution in the expanding de Sitter spacetime in \cite{Nakamura-2014-JMAA} 
into general FLRW spacetimes.
We put 
\beq
\label{Def-G-Gamma}
G:=\frac{1}{Cc}
\left(\frac{a_0^{\mu_0}}{C_0D_{\mu_0}}\right)^{p-1},
\ \  
\gamma:=
\left\{
\frac{2\mu_0}{n(1+\sigma)}-\frac{n-2\mu_0}{2q}
\right\} (p-1)q_\ast,
\eeq
\beq
\label{Def-p1-p2}
p_1(\mu_0):=1+\frac{n(1+\sigma)}{2\mu_0},
\ \ 
p_2(\mu_0):=1+\frac{4}{n-2\mu_0},
\eeq
\beq
\label{Def-J0-J1}
B_0:=\frac{a_0^{\mu_0}}{C_0}\cdot
\left(
\frac{m^\delta}{CcB_1}
\right)^{1/(p-1)},
\ \ 
B_1:=\frac{1}{2H}\cdot \left| \frac{4}{ n(1+\sigma)(\gamma-1) } \right|^{1/q_\ast},
\eeq
\beq
\label{Def-J2-J3}
\ \ 
B_2:=\frac{1}{2H}\left\{\frac{4}{n(1+\sigma)}\right\}^{1/q_\ast},
\ \ 
B_3:=\frac{1}{2H}\left\{\frac{2}{\mu_0(p-1)q_\ast}\right\}^{1/q_\ast},
\eeq
where $C_0$ and $C$ are the constants in Theorem \ref{Thm-13-Cor-14}.

\begin{corollary}[Local and global solutions under \eqref{Def-a} and \eqref{M-FLRW}]
\label{Thm-16-Cor-17}
Assume the conditions on $n$, $\mu_0$, $\mu$, $p$ and $f$ in Theorem \ref{Thm-13-Cor-14}.
Let $m\ge0$.
For $T_0$ and $T_1$ given by \eqref{R-Def-T_0} and \eqref{Def-T1}, 
let $a$ and $M$ be the scale-function and the curved mass given by \eqref{Def-a} and \eqref{M-FLRW}.
Let $T$ satisfy $0<T\le T_1$.
Then the following results hold.

(1) 
(Local solutions.) 
Under one of the following conditions from (i) to (xiii), 
the conditions \eqref{Condition-aM}, \eqref{Condition-T} are satisfied,  
and the results (1) and (2) in Theorem \ref{Thm-13-Cor-14} hold.

(i) 
$H=0$, $\sigma\in \br$, $m>0$, $0<T\le G m^\delta$.

(ii) 
$H>0$, $\sigma\ge0$, $\mu_0>0$, $p>p_1(\mu_0)$, $m>0$, $q_\ast< \infty$, 
$D_{\mu_0}> B_0$,
\[
0<T\le \frac{2}{n(1+\sigma)H}
\left[
\left\{
1-
\left(
\frac{Gm^\delta}{B_1}
\right)^{q_\ast}
\right\}^{-1/(\gamma-1)} -1
\right]. 
\]

(iii)  
$H>0$, $\sigma>0$, $\mu_0>0$, $m=0$, 
$p>p_1(\mu_0)$, $q_\ast<\infty$, 
\beq
\label{Thm-16-3}
1-\left\{
1
+\frac{n(1+\sigma)HT}{2}
\right\}^{1-\gamma}
\le
\left(
\frac{G}{B_1}
\left[
\sigma
\left(\frac{nH}{2c}\right)^2
\left\{
1
+\frac{n(1+\sigma)HT}{2}
\right\}^{-2}
\right]^{\delta/2}
\right)^{q_\ast}.
\eeq

(iv) 
$H>0$, $\sigma\ge0$, $\mu_0=0$, $m>0$, $q_\ast<\infty$, 
\beq
\label{Thm-16-(iv)-1000}
0<T\le \frac{2}{n(1+\sigma)H}
\left[
\left\{
1+
\left(
\frac{Gm^\delta}{B_1}
\right)^{q_\ast}
\right\}^{1/(1-\gamma)}
-1
\right].
\eeq

(v) 
$H>0$, $\sigma\ge0$, $\mu_0>0$, $m>0$, $q_\ast<\infty$, $p<p_1(\mu_0)$ with 
\eqref{Thm-16-(iv)-1000}.

(vi) 
$H>0$, $\sigma>0$, $\mu_0=0$, $m=0$, $q_\ast<\infty$, 
\beq
\label{Thm-16-5}
\left\{
1+\frac{n(1+\sigma)HT}{2}
\right\}^{1-\gamma}
-1
\le
\left(
\frac{G}{B_1}
\left[
\sigma
\left(\frac{nH}{2c}\right)^2
\left\{
1+\frac{n(1+\sigma)HT}{2}
\right\}^{-2}
\right]^{\delta/2}
\right)^{q_\ast}.
\eeq

(vii) 
$H>0$, $\sigma<-1$, $\mu_0>0$, $m> \sqrt{|\sigma|}nH/2c$, $q_\ast<\infty$, 
\beq
\label{Thm-16-6}
1-
\left\{
1+\frac{n(1+\sigma)HT}{2}
\right\}^{1-\gamma}
\le
\left(
\frac{G}{B_1}
\left[
m^2
+
\sigma
\left(\frac{nH}{2c}\right)^2
\left\{
1+\frac{n(1+\sigma)HT}{2}
\right\}^{-2}
\right]^{\delta/2}
\right)^{q_\ast}.
\eeq

(viii) 
$H>0$, $\sigma\ge0$, $\mu_0>0$, $p=p_1(\mu_0)$, $m>0$, $q_\ast<\infty$, 
\[
0<T\le \frac{2}{n(1+\sigma)H}
\left\{
\exp\left(\frac{Gm^\delta}{B_2}\right)^{q_\ast}-1
\right\}.
\]

(ix) 
$H>0$, $\sigma>0$, $\mu_0>0$, $p=p_1(\mu_0)$, $m=0$, $q_\ast<\infty$, 
\beq
\label{Thm-16-8}
\log
\left\{
1+\frac{n(1+\sigma)HT}{2}
\right\}
\le
\left(
\frac{G}{B_2}
\left[
\sigma 
\left(
\frac{nH}{2c}
\right)^2
\left\{
1+\frac{n(1+\sigma)HT}{2}
\right\}^{-2}
\right]^{\delta/2}
\right)^{q_\ast}.
\eeq

(x) 
$H>0$, $\sigma=-1$, $\mu_0>0$, $p>1$, $m>nH/2c$, $q_\ast<\infty$, 
\[
D_{\mu_0}>\frac{a_0^{\mu_0}}{C_0}
\left[
\frac{1}{CcB_3}
\left\{
m^2-\left(\frac{nH}{2c}\right)^2
\right\}^{\delta/2}
\right]^{1/(p-1)},
\]
\[
0<T\le -\frac{1}{\mu_0(p-1)Hq_\ast}
\log
\left(
1-
\left[
\frac{G}{B_3}
\left\{
m^2-\left(\frac{nH}{2c}\right)^2
\right\}^{\delta/2}
\right]^{q_\ast}
\right).
\]

(xi) 
$H>0$, $\sigma=-1$, $\mu_0\ge0$, $m>nH/2c$, $q_\ast<\infty$, 
\[
\begin{cases}
1\le p<\infty & \mbox{if}\ \ \mu_0=0,
\\
p=1 & \mbox{if}\ \ \mu_0>0,
\end{cases}
\]
\[
0<T\le 
\frac{1}{2H}
\left[
2HG
\left\{
m^2-\left(\frac{nH}{2c}\right)^2
\right\}^{\delta/2}
\right]^{q_\ast}.
\]

(xii) 
$H>0$, $\sigma\ge0$, $\mu_0\ge0$, $q_\ast=\infty$, 
\[
p_2(\mu_0)\le p<
\begin{cases}
\infty & \mbox{if}\ \ \mu_0=0, \\
p_1(\mu_0) 
& \mbox{if}\ \ \mu_0>0,
\end{cases}
\]
\beq
\label{Thm-16-11-1}
D_{\mu_0}<\frac{a_0^{\mu_0}}{C_0}
\left[
\frac{H}{Cc}
\left\{
m^2+\sigma\left(\frac{nH}{2c}\right)^2
\right\}^{\delta/2}
\right]^{1/(p-1)},
\eeq
\beq
\label{Thm-16-11-2}
\left\{
1+\frac{n(1+\sigma)HT}{2}\right\}^{\zeta}
\le 
2HG
\left[
m^2+\sigma\left(\frac{nH}{2c}\right)^2
\left\{
1+\frac{n(1+\sigma)HT}{2}
\right\}^{-2}
\right]^{\delta/2},
\eeq
where $\zeta:=1-2\mu_0(p-1)/n(1+\sigma)$.

(xiii) 
$H>0$, $\sigma<-1$, $\mu_0\ge0$, $p\ge p_2(\mu_0)$, $m>\sqrt{|\sigma|}nH/2c$, $q_\ast=\infty$,
\[
0<T
\begin{cases}
<T_1 & \mbox{if}\ \ p=p(\mu_0), 
\\
\le 
-\frac{2}{n(1+\sigma)H}
\left[
1-
\frac{nH}{2c}
\sqrt{ \frac{|\sigma|}{m^2-(2HG)^{-2/\delta}} }
\right]
& \mbox{if}\ \ p<p(\mu_0),
\end{cases}
\]
\[
D_{\mu_0}
\begin{cases}
\le 
\frac{a_0^{\mu_0}}{C_0} 
\left(
\frac{H}{Cc}
\right)^{1/(p-1)}
& \mbox{if}\ \ p=p(\mu_0), 
\\
<
\frac{a_0^{\mu_0}}{C_0} 
\left[
\frac{H}{Cc}
\left\{
m^2+\sigma\left(\frac{nH}{2c}\right)^2
\right\}^{\delta/2}
\right]^{1/(p-1)}
& \mbox{if}\ \ p<p(\mu_0).
\end{cases}
\]

(2) 
(Small global solutions.)
Under one of the following conditions from (i) to (iv), 
the conditions \eqref{Condition-SmallGlobal} are satisfied, 
and the result (3) in Theorem \ref{Thm-13-Cor-14} holds.

(i) 
$H>0$, $\sigma\ge0$, $\mu_0>0$, $p>p_1(\mu_0)$, $m>0$, $q_\ast<\infty$,
\[
D_{\mu_0}\le \frac{a_0^{\mu_0}}{C_0} 
\left(\frac{m^\delta}{Cc B_1}\right)^{1/(p-1)}.
\]

(ii) 
$H>0$, $\sigma=-1$, $\mu_0> 0$, $p>1$, $m>nH/2c$, $q_\ast<\infty$,
\[
D_{\mu_0}\le 
\frac{a_0^{\mu_0}}{C_0} 
\left[
\frac{1}{CcB_3}
\left\{
m^2-\left(\frac{nH}{2c}\right)^2
\right\}^{\delta/2}
\right]^{1/(p-1)}.
\]

(iii) 
$H>0$, $\sigma\ge0$, $\mu_0>0$, $p\ge\max\{p_1(\mu_0),p_2(\mu_0)\}$, $q_\ast=\infty$,
\[
D_{\mu_0}\le \frac{a_0^{\mu_0}}{C_0} 
\left(\frac{Hm^\delta}{Cc}\right)^{1/(p-1)}.
\]

(iv) 
$H>0$, $\sigma=-1$, $\mu_0\ge 0$, $p\ge p_2(\mu_0)$, $q_\ast=\infty$,
\[
D_{\mu_0}\le \frac{a_0^{\mu_0}}{C_0} 
\left[
\frac{H}{Cc}
\left\{
m^2-\left(\frac{nH}{2c}\right)^2
\right\}^{\delta/2}
\right]^{1/(p-1)}.
\]

(3) 
(Large global solutions.)
Let $H\ge0$, $\mu_0=0$, $\lambda\ge0$, $m\ge0$, $\sigma\in \br$ with $m^2+\sigma(nH/2c)^2>0$. 
If $H=0$ or $\sigma\ge-1$, 
then the solution obtained in (1) is a global solution, 
namely, $T=T_0$.
If $H>0$ and $\sigma<-1$,  
then the solution obtained in (1) exists on the interval $[0,T_1)$.
\end{corollary}

\begin{remark}[Fujita exponent]
The conditions for $T$ from (i) to (xiii) in (1) in Corollary \ref{Thm-16-Cor-17} 
are satisfied for sufficiently small $T>0$. 
Global solutions for small data are obtained for $p>1$ in the case (ii) in (2) in Corollary \ref{Thm-16-Cor-17},
which shows that the Fujita exponent greater than $1$ does not appear in the expanding de Sitter spacetime.
It is known that the Cauchy problem of the Klein-Gordon equation in the Minkowski spacetime, i.e., $H=0$, given by 
\beq
\label{Cauchy-Linear}
\begin{cases}
\left(
c^{-2}\partial_t^2-a^{-2}_0\Delta +m^2
\right) u(t,x)
+f(u)(t,x)=0,
\\ 
u(0,\cdot)=u_0(\cdot)\in H^1(\brn),
\ \ 
\partial_t u(0,\cdot)=u_1(\cdot)\in L^2(\brn)
\end{cases}
\eeq
with $a_0>0$, $m>0$, $f(u)=\lambda|u|^p$ or $f(u)=\lambda|u|^{p-1}u$, $\lambda\in \br$, $1<p<\infty$ for $n=1,2$ and  $1<p<1+4/(n-2)$ for $n\ge3$ 
admits global solutions for any real-valued small data  $u_0$ and $u_1$.
Indeed, 
the energy defined by 
$E_\ast(u)(t):=c^{-2}\|\partial_tu(t)\|_2^2+a_0^{-2}\|\nabla u(t)\|_2^2+m^2\|u(t)\|_2^2$
satisfies 
\[
E_\ast(u)(t)+\int_\brn V(u)(t,x)dx=E_\ast(u)(0)+\int_\brn V(u)(0,x)dx,
\]
where $V(u)=2\lambda|u|^{p+1}/(p+1)$ for $f(u)=\lambda |u|^{p-1}u$, 
and $V(u)=2\lambda|u|^{p}u/(p+1)$ for $f(u)=\lambda |u|^{p}$ since $u$ is real-valued.
By the bound 
\[
\left|\int_\brn V(u)dx\right|
\lesssim 
\frac{|\lambda|}{p+1}
\left(\|u\|^{1-\theta}_2\|\nabla u\|_2^\theta\right)^{p+1}
\lesssim
\frac{|\lambda|\left(m^{-1+\theta}a_0^{\theta}\right)^{p+1}}{p+1}E_\ast(u)^{(p+1)/2},
\]
where $\theta:=n(p-1)/2(p+1)$, 
we have the inequality
\[
E_\ast(u)(t)\le E_\ast(u)(0)
+
\frac{C|\lambda|\left(m^{-1+\theta}a_0^{\theta}\right)^{p+1}}{p+1}
\left\{ E_\ast(u)(0)^{(p+1)/2}+E_\ast(u)(t)^{(p+1)/2}\right\}
\]
for some constant $C>0$, 
which shows the uniform bound of $E_\ast(u)(t)$ for $t\ge 0$ for sufficiently small data, 
and we obtain global solutions since the existence-time of local solutions depends only on the norms of initial data.
We refer to \cite[p. 631]{Keel-Tao-1999-AJM} on the argument of the Fujita exponent on the Klein-Gordon equation with combined nonlinear terms.
On the other hand, the above argument does not work for our problem \eqref{Cauchy} 
since general FLRW spacetimes yield dissipative or anti-dissipative terms in the energy estimates (see \eqref{Lem-2-1000} and \eqref{Lem-2-5}, below). 
\end{remark}

We have considered local and global solutions in Theorem \ref{Thm-13-Cor-14} and Corollary \ref{Thm-16-Cor-17} 
in the case of the spatial expansion $\dot{a}\ge0$.
Especially, the condition $\dot{a}>0$ yields the dissipative property on the energy estimate of the equation 
(see the term $\|\sqrt{a^{-3}\dot{a}} \nabla u\|_{L^2L^2}$ 
of (2) in Lemma \ref{Lem-2}, below). 
By this property, we obtain the small global solutions, 
while it does not work when $\dot{a}< 0$.

Next, we show blowing-up solutions for the Cauchy problem \eqref{Cauchy} in the case of spatial contraction $\dot{a}\le 0$ with  gauge-invariant nonlinear terms.
We are based on the argument in \cite[Theorem 7]{Menzala-Strauss-1982-JDE} for the Klein-Gordon equation 
with Hartree-type nonlinear term in the Minkowski spacetime to prove Theorem \ref{Thm-19}  
(see \cite[Theorem 5]{Nakamura-Takashima-2021-DIE} for the Hartree-type in the de Sitter spacetime).
The excluded case $\dot{a}>0$ in Theorem \ref{Thm-19} with $\dot{M}\le 0$ and $m>0$ has been considered by McCollum, Mwamba and Oliver in 
\cite[Theorem 1]{McCollum-Mwamba-Oliver-2024-NA}.

\begin{theorem}[Blowing-up under the gauge-invariance]
\label{Thm-19}
Let $T_0$, $T_1$, $a$ and $M$ satisfy 
\beq
\label{Condition-aM-BlowUp}
\begin{cases}
0<T_1\le T_0\le \infty,
\\
a\in C^2([0,T_0),(0,\infty)),
\ \ 
\dot{a}\le0, 
\\ 
M\in C^1([0,T_1),[0,\infty)), 
\ \ 
\dot{M}\le0.
\end{cases}
\eeq
Put $a_0:=a(0)$, $a_1:=\dot{a}(0)$ and $M_0:=M(0)$.
Let $n$, $p$, $f$, $\kappa$, $\kappa_\ast$ satisfy  
\beq
\label{Thm-19-1000}
n\ge1, \ \ 
\begin{cases}
1+\frac{4}{n}\le p<\infty & \mbox{if}\ \  \dot{a}\not\equiv0, \\
1< p<\infty & \mbox{if}\ \  \dot{a}\equiv0, \\
\end{cases}
\eeq
\beq
\label{Thm-19-2000}
\lambda<0, 
\ f(u):=\lambda |u|^{p-1}u, 
\ 2< \kappa\le p+1, 
\ 0<\kappa_\ast< \frac{(\kappa-2)}{4}.
\eeq
Let $u_0\in H^1(\brn)$ and $u_1\in L^2(\brn)$ satisfy 
\beq
\label{Thm-19-4000}
\frac{\|u_1\|_{L^2(\brn)}^2}{c^{2}}
+
\frac{\|\nabla u_0\|_{L^2(\brn)}^2}{a_0^{2}}
+
M_0^2\|u_0\|_{L^2(\brn)}^2
+
\frac{2\lambda\|u_0\|_{L^{p+1}(\brn)}^{p+1}}{a_0^{n(p-1)/2}(p+1)}<0,
\eeq
\beq
\label{Thm-19-5000}
a_1 \|u_0\|_{L^2(\brn)}^2+a_0\Re \int_\brn u_0(x)\overline{u_1(x)}dx>0,
\eeq
\beq
\label{Thm-19-6000}
T_\ast
:=
\frac{1}{2\kappa_\ast}
\cdot
\frac{a_0 \|u_0\|_{L^2(\brn)}^2}
{a_1\|u_0\|_{L^2(\brn)}^2+a_0 \Re \int_\brn u_0(x)\overline{u_1(x)}dx}\le T_1.
\eeq
Assume 
\beq
\label{Thm-19-3000}
(\kappa-2)
\left\{
c^2M^2
+
\left(
\frac{ \dot{a} }{a}
\right)^2
\right\}
+
2\frac{d}{dt}
\left(
\frac{ \dot{a} }{a}
\right)
\ge0
\eeq
on $[0,T_\ast)$.
Then the solution $u$ of the Cauchy problem \eqref{Cauchy} blows-up in finite time in $L^2(\brn)$ 
no later than $T_\ast$. 
\end{theorem}

\begin{remark}
The assumption \eqref{Thm-19-4000} or \eqref{Thm-19-5000} requires $u_0\not\equiv0$. 
The assumption \eqref{Thm-19-4000} is satisfied for negatively large $\lambda<0$ when $u_0\not\equiv0$.
The assumption \eqref{Thm-19-5000} is satisfied when $u_0\not\equiv0$ and $u_1=\rho u_0$ with a real number $\rho$ sufficiently large. 
\end{remark}

The following corollary is derived from Theorem \ref{Thm-19} 
for the scale-function $a$ defined by \eqref{Def-a}.
We refer to \cite[Theorem 5.1]{ZhangJian-2002-NA} 
for blowing-up solutions in the Minkowski spacetime, i.e., $H=0$, with gauge-invariant nonlinear terms 
with the initial energy controlled by a ground state of stationary elliptic equation.
We refer to \cite[Proposition 1.1]{Nakamura-2021-JMP} and \cite[Theorems 1.1 and 1.2]{Yagdjian-2009-DCDS}  
for blowing-up solutions in the de Sitter spacetime, i.e., $\sigma=-1$, for gauge-variant nonlinear terms based on different techniques due to Kato's inequality.
There are several papers on blowing-up solutions in FLRW spacetimes with gauge-variant nonlinear terms 
(see \cite{Wei-Yong-2024-JMP} and the references therein).

Put 
\beq
\label{Def-T2}
T_2:=
-\frac{2}{n(1+\sigma)H}
\left[
1+
\frac{H}{mc}\cdot 
\left\{
\frac{n(1+\sigma)-(p-1)\left(\sigma n^2/4+1\right) }{p-1}
\right\}^{1/2}
\right]
\eeq
\beq
\label{Def-pAst-pSharp}
p_\ast:=1+\frac{4n(1+\sigma)}{4+\sigma n^2},
\ \ 
p_{\sharp}:=1+\frac{4n(1+\sigma)}{4+\sigma n^2+(2mc/H)^2},
\eeq
where we put $T_2=\infty$ when $(1+\sigma)H=0$.

\begin{corollary}[Blowing-up  under the gauge-invariance under \eqref{Def-a} and \eqref{M-FLRW}]
\label{Cor-20}
Let $n$, $m$, $p$, $f$, $\kappa$, $\kappa_\ast$ satisfy \eqref{Thm-19-1000} and \eqref{Thm-19-2000} with $\kappa=p+1$.
Let $T_0$, $T_1$, $T_2$ and $T_\ast$ be defined by \eqref{R-Def-T_0}, \eqref{Def-T1}, \eqref{Def-T2} and \eqref{Thm-19-6000}.
Let $u_0\in H^1(\brn)$ and $u_1\in L^2(\brn)$ satisfy \eqref{Thm-19-4000}, \eqref{Thm-19-5000} and \eqref{Thm-19-6000}.
Assume one of the following conditions from (i) to (vii) holds.

\vspace{5pt}

(i) $H=0$, $\sigma\in \br$, $m\ge0$.

(ii) $H<0$, $\sigma=-1$, $m\ge \frac{n|H|}{2c}$.

(iii) $H<0$, $\sigma=0$, $m\ge0$, $p\ge p_\ast$, $T_\ast\le T_0$.

(iv) $H<0$, $\sigma=0$, $m\ge0$, $p_\sharp<p< p_\ast$, $T_\ast\le T_2$.

(v) $H<0$, $\max\{-1,-\frac{4}{n^2} \}<\sigma<0$, $p\ge p_\ast$, $m>\frac{\sqrt{|\sigma|} n|H|}{2c}$, $T_\ast\le T_1$.

(vi) $H<0$, $\max\left[-1, -\frac{4}{n^2}\left\{1+\left(\frac{mc}{H}\right)^2\right\}\right]<\sigma\le -\frac{4}{n^2}$, $p>p_\sharp$, $m>\frac{\sqrt{|\sigma|} n|H|}{2c}$, $T_\ast\le \min\{T_1,T_2\}$.

(vii) $H<0$, $\max\left\{-1,-\frac{4}{n^2}\right\}<\sigma<0, p_\sharp<p<p_\ast$, $m>\frac{\sqrt{|\sigma|} n|H|}{2c}$, $T_\ast\le \min\{T_1,T_2\}$.

\vspace{5pt}

\noindent
Then the conditions \eqref{Condition-aM-BlowUp} and \eqref{Thm-19-3000} are satisfied, and the result in Theorem \ref{Thm-19} hold.
Namely, the solution $u$ of the Cauchy problem \eqref{Cauchy} blows-up in finite time in $L^2(\brn)$ no later than $T_\ast$.
\end{corollary}

We note 
$T_0=T_1=\infty$ for (i) and (ii), 
$T_0=T_1=-2/n(1+\sigma)H$ for (iii) and (iv),
$T_0=-2/n(1+\sigma)H>T_1>0$ for (v), (vi) and (vii)  
in Corollary \ref{Cor-20}.

In Corollaries \ref{Thm-16-Cor-17} and \ref{Cor-20}, 
the case $H=0$ reduces to the semilinear Klein-Gordon equation in the Minkowski spacetime, 
which was extensively studied 
(see, e.g., 
\cite{Nakamura-Ozawa-2001-PRIMS, Nakamura-Ozawa-2002-ASNSP} 
for closely related results on global solutions, 
and 
\cite{Ball-1978-AcademicPress},
\cite{Levine-1974-TransAMS},
\cite{Yang-Xu-2018-ApplMathLetters} 
on blowing-up solutions, and the references therein). 
We included this case to compare it with the case $H\neq0$.

We denote the inequality $A\le CB$ for some constant $C>0$ which is not essential for the argument by $A\lesssim B$.
This paper is organized as follows.
In Section \ref{Sec-Pre}, 
we collect several results on energy estimates, fundamental solutions and evolution operators for linear equations, 
the scale-function, and the curved mass in FLRW spacetimes, 
which are required to prove 
Theorem \ref{Thm-13-Cor-14},  
Corollary \ref{Thm-16-Cor-17},
Theorem \ref{Thm-19},
Corollary \ref{Cor-20} 
in Sections \ref{Sec-Thm-13-Cor-14},
\ref{Sec-Thm-16-Cor-17},
\ref{Sec-Thm-19},
\ref{Sec-Cor-20}, respectively.


\newsection{Preliminaries}
\label{Sec-Pre}
We prepare several estimates to prove the results in the previous section.
For $T$ with $0<T\le T_0$, 
let us consider the Cauchy problem of the linear equation 
\beq
\label{Cauchy-Linear}
\begin{cases}
\left(
c^{-2}\partial_t^2-a^{-2}(t)\Delta +M^2(t)
\right) u(t,x)
+h(t,x)=0,
\\ 
u(0,x)=u_0(x),
\ \ 
\partial_t u(0,x)=u_1(x)
\end{cases}
\eeq
for $0\le t<T$ and $x\in \brn$, 
where $h$ is any function on $[0,T)\times \brn$.
We collect fundamental energy estimates for the Cauchy problem with their proofs for the completeness of the paper.
We recall that $M$ takes purely imaginary values when $M^2<0$.


\begin{lemma}[Energy estimates]
\label{Lem-2}
Let $n\ge1$.
Let $T_0$, $a$ and $M$ satisfy 
\beq
\label{Condition-aM-Energy}
\begin{cases}
0<T_0\le \infty,
\\
a\in C^2([0,T_0),(0,\infty)),
\\ 
M^2\in C^1([0,T_0),\br).
\end{cases}
\eeq
Put $a_0:=a(0)$ and $M_0:=M(0)$.
For any $0<T\le T_0$ and the solution $u$ of \eqref{Cauchy-Linear} on $[0,T)$, the following results hold.

(1) 
Put $e^0:=c^{-2} |\partial_tu|^2+a^{-2}|\nabla u|^2+M^2 |u|^2$, 
$(e^1,\cdots,e^n):=-a^{-2} 2\Re(\overline{\partial_tu} \nabla u)$,
$e^{n+1}:=-2M\dot{M} |u|^2+2a^{-3} \dot{a} |\nabla u|^2$.
Then 
\beq
\label{Lem-1-Proof-50}
\partial_t e^0+\sum_{j=1}^n \partial_j e^j+e^{n+1}+ 2\Re( \overline{\partial_t u} h)=0
\eeq
and 
\begin{multline}
\label{Lem-2-1000}
E(u)(t)
:=
c^{-2} \|\partial_t u(t)\|_2^2+a^{-2}(t)\|\nabla u(t)\|_2^2+M^2(t)\|u(t)\|_2^2
-2\int_0^t M\dot{M}\|u\|_2^2 ds
\\
+2\int_0^t a^{-3}\dot{a}\|\nabla u\|_2^2 ds
+2\Re\int_0^t \int_\brn \overline{\partial_t u} h dx ds
=
E(u)(0).
\end{multline}

(2) 
If $\dot{a}\ge0$, $M^2\ge0$ and $M\dot{M}\le 0$, 
then  
\begin{multline}
\label{Ineq-EnergyEs-Formal}
c^{-1}\|\partial_tu\|_{L^\infty L^2}
+\|a^{-1}\nabla u\|_{L^\infty L^2}
+\|M u\|_{L^\infty L^2}
+\left\|\sqrt{-M\dot{M}} u\right\|_{L^2 L^2}
+\left\|\sqrt{a^{-3} \dot{a} } \nabla u\right\|_{L^2 L^2}
\\
\lesssim  
c^{-1}\|u_1\|_{2}
+a_0^{-1}\|\nabla u_0\|_{2}
+\|M_0 u_0\|_{2}
+c\|h\|_{L^1L^2},
\end{multline}
where $L^pL^2:=L^p((0,T),L^2(\brn))$ for $p=1,2,\infty$.

(3) 
Put 
$\widetilde{e}^0:=e^0+a^{-n(p-1)/2} V(u)$,
$\widetilde{e}^{n+1}:=e^{n+1}+{n(p-1)} \dot{a}V(u)/2a^{n(p-1)/2+1}$,
where 
\beq
\label{f-V}
\begin{cases}
f(u):=\lambda |u|^{p-1}u, \ \ \lambda\in \br, \ \ 1\le p<\infty,
\\ 
V(u):=\frac{ 2\lambda|u|^{p+1} }{p+1},
\ \ h:=a^{n/2}f(a^{-n/2}u). 
\end{cases}
\eeq
Then 
\beq
\label{Lem-2-2000}
\partial_t \widetilde{e}^0+\sum_{j=1}^n \partial_j e^j+\widetilde{e}^{n+1}=0
\eeq
and 
\begin{multline}
\label{Lem-2-5}
\widetilde{E}(u)(t)
:=
c^{-2} \|\partial_tu(t)\|_2^2
+
a^{-2}(t) \|\nabla u(t)\|_2^2
+
M^2(t)\|u(t)\|_2^2
+
a^{-n(p-1)/2}(t) \int_\brn V(u)(t)dx
\\
-
2\int_0^t M\dot{M}\|u\|_2^2ds
+
2\int_0^t \frac{\dot{a}}{a^{3}}(s) \|\nabla u(s)\|_2^2 ds
\\
+
\frac{n(p-1)}{2}\int_0^t \frac{\dot{a}}{a^{n(p-1)/2+1 }} (s)\int_\brn V(u)(s,x)dx ds
=
\widetilde{E}(u)(0).
\end{multline}

(4) We have  
\begin{multline}
\label{Lem-2-3000}
\partial_t^2|u|^2-2c^2a^{-2}\Re\nabla(\overline{u}\nabla u)-2|\partial_tu|^2
+2c^2a^{-2}|\nabla u|^2
+2c^2M^2|u|^2
\\
+2c^2\Re\left(\overline{u} h\right)=0
\end{multline}
and 
\begin{multline}
\label{Lem-2-4000}
\partial_t^2\|u\|^2_2-2\|\partial_tu\|^2_2
+2c^2a^{-2}\|\nabla u\|^2_2
+2c^2M^2\|u\|_2^2
+2c^2\Re \int_\brn \overline{u} h dx=0.
\end{multline}

(5) 
Under \eqref{f-V}, we have 
\begin{multline}
\label{Lem-2-8}
\frac{d^2}{dt^2}\|u\|_2^2-2\|\partial_t u\|_2^2+2c^2a^{-2}\|\nabla u\|_2^2+2c^2M^2\|u\|_2^2
+
2\lambda c^2a^{-n(p-1)/2} \|u\|_{p+1}^{p+1}=0.
\end{multline}
\end{lemma}


\begin{proof}
(1) 
Multiplying $\overline{\partial_t u}$ to the first equation in \eqref{Cauchy-Linear}, and taking its real part, we have the divergence form \eqref{Lem-1-Proof-50}.
Integrating both sides of \eqref{Lem-1-Proof-50}, we obtain \eqref{Lem-2-1000}.

(2)
By \eqref{Lem-2-1000}, we have 
\begin{multline*}
c^{-2}\|\partial_t u\|_{L^\infty((0,T),L^2(\brn))}^2
+
\|a^{-1}\nabla u\|_{L^\infty((0,T),L^2(\brn))}^2
+
\|Mu\|_{L^\infty((0,T),L^2(\brn))}^2
\\
+
\int_0^T \int_\brn e^{n+1} dx ds
\le 4E(u)(0)+8\int_0^T \int_\brn |\overline{\partial_t u} h| dx ds.
\end{multline*}
We obtain the required result by
\[
8\int_0^T \int_\brn |\overline{\partial_t u} h| dx ds
\le (2c^2)^{-1}\|\partial_t u\|_{L^\infty((0,T),L^2(\brn))}^2+32c^2\|h\|_{L^1((0,T),L^2(\brn))}^2.
\]

(3) 
By the assumption \eqref{f-V}, we have 
\beq
\label{Proof-Lem-2-500}
2\Re\left(\overline{\partial_t u} u\right)=\partial_t |u|^2,
\ \ 
2\Re\left\{\overline{\partial_t u} f(u)\right\}=\partial_t V(u),
\ \ 
\overline{u} f(u)=\lambda|u|^{p+1}
\eeq
and 
\[
2\Re(\overline{\partial_tu}h)
=
\partial_t(a^{-n(p-1)/2} V)+\frac{n(p-1)\dot{a} V}{2a^{n(p-1)/2+1}}.
\]
So that, the equation \eqref{Lem-1-Proof-50} is rewritten by 
\beq
\label{Proof-Lem-2-1000}
\partial_t(e^0+a^{-n(p-1)/2} V)+\sum_{j=1}^n \partial_j e^j+e^{n+1}+\frac{n(p-1)\dot{a} V}{2a^{n(p-1)/2+1}}=0,
\eeq
which is \eqref{Lem-2-2000}.
Integrating both sides, we obtain \eqref{Lem-2-5}.

(4)
Multiplying $\overline{u}$ to the first equation in \eqref{Cauchy-Linear} and taking its real part, we have 
\[
2\Re(\overline{u} \partial_t^2u)
-2c^2a^{-2}\Re(\overline{u} \Delta u)
+2c^2M^2|u|^2
+2c^2\Re(\overline{u} h)=0
\]
which is rewritten by \eqref{Lem-2-3000}, 
where we have used 
$2\Re(\overline{u} \partial_t^2u)=\partial_t^2 |u|^2-2|\partial_t u|^2$, and
$\overline{u} \Delta u=\nabla \left( \overline{u} \nabla u\right)-|\nabla u|^2$.

(5) The required result follows from \eqref{Lem-2-4000} and 
$\overline{u} h=\lambda a^{-n(p-1)/2} |u|^{p+1}$.
\end{proof}

\vspace{10pt}

We derive the integral equation corresponding to \eqref{Cauchy-Linear}.
The differential equation in \eqref{Cauchy-Linear} is transformed to
\beq
\label{Lem-8}
\partial_t^2 Fu(t,\xi)+\alpha(t,\xi)Fu(t,\xi)+c^2Fh(t,\xi)=0
\eeq
for $\xi\in \brn$ by the Fourier transform $F$,
where $\alpha$ is defined by 
\beq
\label{Def-Alpha}
\alpha(t)=\alpha(t,\xi):=\frac{c^2|\xi|^2}{a^2(t)}+c^2M^2(t).
\eeq
Let $\rho_0$ and $\rho_1$ be the solutions of the differential equation 
\beq
\label{Def-Rho-j}
\partial_t^2\rho_j(t,\xi)+\alpha(t,\xi) \rho_j(t,\xi)=0,
\ \ 
\rho_j(0,\xi)=\delta_{j0},
\ \ 
\partial_t \rho_j(0,\xi)=\delta_{j1},
\eeq
where $\delta_{jk}$ denotes Kronecker's delta, i.e., 
$\delta_{00}=\delta_{11}=1$ and $\delta_{01}=\delta_{10}=0$. 
We have the following elementary properties of $\rho_0$ and $\rho_1$,
we prove them for completeness.

\begin{lemma}
\label{Lem-4}
The functions $\rho_0$ and $\rho_1$ with \eqref{Def-Rho-j} satisfy the followings.

(1) 
$\det
\begin{pmatrix}
\rho_0 & \rho_1 \\
\partial_t \rho_0 & \partial_t \rho_1 
\end{pmatrix}
=1$.

(2) 
If $\alpha$ satisfies $\alpha\ge 0$ and $\partial_t{\alpha}\le 0$, 
then the following estimates hold.
\beq
\label{Lem-4-1000}
|\rho_0(t,\cdot)|\le \sqrt{\frac{\alpha(0,\cdot)}{\alpha(t,\cdot)} },
\ \ 
|\partial_t\rho_0(t,\cdot)|\le \sqrt{\alpha(0,\cdot)},
\ \ 
|\rho_1(t,\cdot)|\le\frac{1}{ \sqrt{\alpha(t,\cdot)} },
\ \ 
|\partial_t\rho_1|\le 1.
\eeq
\end{lemma}

\begin{proof}
(1) 
The result follows from 
$\partial_t \det
\begin{pmatrix}
\rho_0 & \rho_1 \\
\partial_t \rho_0 & \partial_t \rho_1 
\end{pmatrix}
=\rho_0\partial_t^2\rho_1
-\rho_1\partial_t^2\rho_0
=0$
by \eqref{Def-Rho-j}.

(2) 
Multiplying $\overline{\partial_t\rho_j}$ to the differential equation in \eqref{Def-Rho-j}, 
and taking its real part, we have 
\[
\partial_t
\left(|\partial_t \rho_j|^2+\alpha|\rho_j|^2\right)-\partial_t\alpha |\rho_j|^2=0
\]
for $j=0,1$.
Integrating the both sides, we have 
\[
|\partial_t \rho_j(t)|^2+\alpha(t)|\rho_j(t)|^2-\int_0^t \partial_t \alpha(s)|\rho_j(s)|^2 ds
=
|\partial_t\rho_j(0)|^2+\alpha(0)|\rho_j(0)|^2,
\]
where $\rho_j(t)=\rho_j(t,\cdot)$.
So that, we obtain the required results by the initial conditions in \eqref{Def-Rho-j} under $\alpha>0$ and $\partial_t \alpha\le0$.
\end{proof}

For any given functions $A$ and $B$ on $\brn$, 
the solution of the Cauchy problem 
\beq
\label{ODE-Rho}
\partial_t^2\rho(t,\xi)+\alpha(t,\xi)\rho(t,\xi)+\beta(t,\xi)=0,
\ \ 
\rho(0,\xi)=A(\xi),
\ \ 
\partial_t \rho(0,\xi)=B(\xi)
\eeq
is represented by  
\beq
\label{Def-Rho}
\rho(t,\xi):=\rho_0(t,\xi)A(\xi)+\rho_1(t,\xi)B(\xi)-\int_0^t\rho_2(t,s,\xi) \beta(s,\xi)ds,
\eeq
where we have put  
\[
\rho_2(t,s,\xi):=\rho_1(t,\xi)\rho_0(s,\xi)-\rho_0(t,\xi)\rho_1(s,\xi),
\]
and the differential equation \eqref{ODE-Rho} follows from \eqref{Def-Rho} by (1) in Lemma \ref{Lem-4}.
Putting $\rho:=Fu$, $\beta:=c^2Fh$, $A:=Fu_0$ and $B:=Fu_1$, 
the solution $u$ of \eqref{Cauchy-Linear} 
is represented by the integral equation 
\beq
\label{Cauchy-Linear-Int}
u(t)=K_0(t)u_0+K_1(t) u_1-c^2\int_0^t K_2(t,s) h(s) ds,
\eeq
where $K_0$, $K_1$ and $K_2$ are the operators defined by 
\beq
\label{Def-K}
K_0(t):=F^{-1}\rho_0(t) F,
\ \ 
K_1(t):=F^{-1}\rho_1(t) F,
\ \ 
K_2(t,s):=F^{-1}\rho_2(t,s) F.
\eeq

To show fundamental estimates for the operators $K_0$, $K_1$ and $K_2$ in Lemma \ref{Lem-7}, below,
we prepare the following two lemmas.
Put $\langle \xi \rangle:=\sqrt{1+|\xi|^2}$ for $\xi\in \brn$.

\begin{lemma}
\label{Lem-5}
For $T>0$, let $a\in C^2([0,T),(0,\infty))$ and $M\in C^1([0,T),(0,\infty))$ 
with $\dot{a}\ge0$, $\dot{M}\le 0$, $a_0:=a(0)$, and $M_0:=M(0)$.
Put 
\beq
\label{Def-Eta}
\eta(t):=\frac{M_0a(t)}{M(t)a_0}
\eeq
for $0\le t<T$.
Put $M_\ast:=\inf_{0\le t<T} M$.
Then $\alpha$ defined by \eqref{Def-Alpha} satisfies the followings on $[0,T)\times \brn$. 

(1) 
\ \ 
$c\min\{a_0^{-1},M_0\} \langle \xi \rangle\le \sqrt{\alpha(0,\xi)}\le c\max\{a_0^{-1},M_0\} \langle \xi \rangle$.

(2) 
\ \ 
$\partial_t\alpha\le 0$.

(3)
\ \ 
$\eta(t)\ge \sqrt{\frac{\alpha(0,\xi)}{\alpha(t,\xi)} }$.

(4) 
\ \ 
$\sqrt{\alpha(t,\xi)}\ge cM_\ast$.

(5) 
\ \ 
$\sqrt{ \frac{\alpha(0,\xi)}{\alpha(t,\xi)} } 
\le 
\min\left\{
\eta(t), \frac{ \max\{a_0^{-1},M_0\}\langle \xi \rangle}{M_\ast}
\right\}$.

(6) 
\ \ 
$\frac{1}{\sqrt{\alpha(t,\xi)}} 
\le 
\frac{1}{c\min\{a_0^{-1},M_0\}}
\min
\left\{
\frac{\eta(t)}{ \langle\xi\rangle}, 
\frac{ \max\{a_0^{-1},M_0\}}{M_\ast } 
\right\}$.

(7) 
\ \ 
$\dot{\eta}\ge0$. 
\end{lemma}

\begin{proof}
(1)
We have 
$\alpha(0)\le c^2\max\{a_0^{-2},M_0^2\}\langle \xi\rangle^2$ 
by 
$\alpha(0)\le c^2a_0^{-2}(|\xi|^2+1)$ when $a_0^2M_0^2\le 1$, 
and 
$\alpha(0)\le c^2M_0^{2}(|\xi|^2+1)$ when $a_0^2M_0^2\ge 1$.
We also have 
$\alpha(0)\ge c^2\min\{a_0^{-2},M_0^2\}\langle \xi\rangle^2$ 
by 
$\alpha(0)\ge c^2a_0^{-2}(|\xi|^2+1)$ when $a_0^2M_0^2\ge 1$, 
and 
$\alpha(0)\ge c^2M_0^{2}(|\xi|^2+1)$ when $a_0^2M_0^2\le 1$.
So that, we obtain the required result.

(2)
The result follows from $\partial_t\alpha=-2c^2a^{-3} \dot{a} |\xi|^2+2c^2M \dot{M}\le 0$ 
by   
$\dot{a}\ge0$ and $\dot{M}\le 0$.

(3) 
We have 
\[
\alpha(0,\xi)
\le 
\frac{c^2M_0^2}{M^2}
\left\{
\frac{|\xi|^2}{a_0^2} +M^2
\right\}
\le 
\frac{M_0^2a^2}{M^2a_0^2}
\,
\alpha(t,\xi)
\]
by $M\le M_0$ and $a_0\le a$, 
which yields the required result.

(4) 
The result easily follows from $\alpha(t,\xi)\ge c^2M^2$.

(5) 
The result follows from (1), (3), (4) and the definition of $\eta$.

(6) 
The result follows from (1) and (5).

(7) 
Since we have  
\[
\frac{d \eta^2}{dt}
=
\frac{2M_0^2 a^2}{M^2 a_0^2} 
\left(
\frac{\dot{a}}{a}-\frac{\dot{M}}{M}
\right)
\ge0
\]
by $\dot{a}\ge0$ and $\dot{M}\le0$, 
we obtain the result $\dot{\eta}\ge0$.
\end{proof}

\begin{lemma}
\label{Lem-6}
Under the assumption in Lemmas \ref{Lem-4} and \ref{Lem-5},
put
\beq
\label{Def-N}
\begin{array}{ll}
N_1
:=
\frac{\max\{a_0^{-1},M_0\} }{M_\ast},
& 
N_2
:=
\max\{a_0^{-1},M_0\},
\\ 
N_3
:=
\frac{1}{\min\{a_0^{-1},M_0\} },
& 
N_4
:=
\frac{\max\{a_0^{-1},M_0\} }
{M_\ast \min\{a_0^{-1},M_0\} }.
\end{array}
\eeq
Then $\rho_0$ and $\rho_1$ defined by \eqref{Def-Rho-j} satisfy the followings for $t\ge0$ and $\xi\in \brn$.

(1)
$|\rho_0(t,\xi)|\le \min\{\eta(t),N_1\langle \xi\rangle\}$.

(2) 
$|\partial_t\rho_0(t,\xi)|\le cN_2\langle \xi\rangle$.

(3) 
$|\rho_1(t,\xi)|\le \frac{1}{c}\min\left\{\frac{N_3\eta(t)}{\langle \xi\rangle}, N_4\right\}$.

(4) 
$|\partial_t\rho_1|\le 1$.
\end{lemma}

\begin{proof}
We are able to apply Lemma \ref{Lem-4} since $\alpha\ge 0$ and $\partial_t\alpha\le 0$ hold by the definition of $\alpha$ in \eqref{Def-Alpha} and (2) in Lemma \ref{Lem-5}.
The result (1) follows from \eqref{Lem-4-1000} and (5) in Lemma \ref{Lem-5}.
Similarly, the results (2), (3) follow from (1), (6) in Lemma \ref{Lem-5}, respectively.
The result (4) is from \eqref{Lem-4-1000} 
\end{proof}

We show fundamental estimates for the operators $K_0$, $K_1$ and $K_2$ defined by \eqref{Def-K}.

\begin{lemma}
\label{Lem-7}
Under the assumption in Lemma \ref{Lem-6},
the operators $K_0$, $K_1$ and $K_2$ satisfy the followings.
\begin{eqnarray*}
&(1)& 
\|K_0(t)\phi\|_{L^2(\brn)}\le \min\left\{\eta(t)\|\phi\|_{L^2(\brn)},N_1\|\phi\|_{H^1(\brn)}\right\}.
\\
&(2)& 
\|\partial_t K_0(t)\phi\|_{L^2(\brn)}\le cN_2\|\phi\|_{H^1(\brn)}.
\\
&(3)& 
\|K_1(t)\phi\|_{L^2(\brn)}\le \frac{1}{c}\min\left\{N_3\eta(t)\|\phi\|_{H^{-1}(\brn)},N_4\|\phi\|_{L^2(\brn)}\right\}.
\\
&(4)& 
\|\partial_t K_1(t)\phi\|_{L^2(\brn)}\le \|\phi\|_{L^2(\brn)}.
\\
&(5)& 
\|K_1(t)K_0(s)\phi\|_{L^2(\brn)}
\le 
\frac{1}{c}\min\left\{
N_3\eta(t)\eta(s)\|\phi\|_{H^{-1}(\brn)},
N_1N_3\eta(t)\|\phi\|_{L^2(\brn)},\right.
\\
&&
\ \hspace{6cm}\left. 
N_4\eta(s)\|\phi\|_{L^2(\brn)},
N_1N_4\|\phi\|_{H^1(\brn)}
\right\}.
\\
&(6)& 
\|\partial_t K_1(t)K_0(s)\phi\|_{L^2(\brn)}
\le 
\min\left\{
\eta(s)\|\phi\|_{L^2(\brn)},
N_1\|\phi\|_{H^1(\brn)} 
\right\}.
\\
&(7)& 
\|\partial_t K_0(t)K_1(s)\phi\|_{L^2(\brn)}
\le 
\min\left\{
N_2N_3\eta(s)\|\phi\|_{L^2(\brn)},
N_2N_4\|\phi\|_{H^1(\brn)}
\right\}.
\\
&(8)& 
\|K_2(t,s)\phi\|_{L^2(\brn)}
\le 
\frac{2}{c}\min\left\{
N_3\eta(t)\eta(s)\|\phi\|_{H^{-1}(\brn)},
\max\{N_1N_3,N_4\} \eta(t)\|\phi\|_{L^2(\brn)},\right.
\\
&&
\ \hspace{5cm}
\left. 
\max\{N_1N_3,N_4\}\eta(s)\|\phi\|_{L^2(\brn)},
N_1N_4\|\phi\|_{H^1(\brn)}
\right\}.
\\
&(9)& 
\|\partial_t K_2(t,s)\phi\|_{L^2(\brn)}
\le 
2\min
\left\{
\max\{1,N_2N_3\}\eta(s)\|\phi\|_{L^2(\brn)},
\right.
\\
&&
\ \hspace{5cm}
\left.
\max\{N_1,N_2N_4\}\|\phi\|_{H^1(\brn)}
\right\}.
\end{eqnarray*}
\end{lemma}

\begin{proof}
The result (1) follows from the $L^2$-invariance of the Fourier transform and (1) in Lemma \ref{Lem-6}.
Similarly, (2), (3) and (4) follow from (2), (3) and (4) in Lemma \ref{Lem-6}, respectively.
The result (5) follows from the combination of (1) and (3).
Similarly, (6) and (7) follow from the combination of (1) and (4), and the combination of (2) and (3), respectively.
The result (8) follows from (5), and the result (9) follows from (6) and (7).
\end{proof}

We prepare the following properties of $a$ and $M$ defined by \eqref{Def-a} and \eqref{M-FLRW}.

\begin{lemma}
\label{Lem-10}
Let $H\in \br$ and $\sigma\in \br$.
Let $T_0$ and $a$ be given by \eqref{R-Def-T_0} and \eqref{Def-a}.
Then the following results hold.
\begin{eqnarray*}
(1) 
&&
\frac{\dot{a}}{{a}}=H\left(\frac{a}{a_0}\right)^{-n(1+\sigma)/2}.
\\
(2) 
&&
\frac{\ddot{a}}{{a}}=H^2\left(\frac{a}{a_0}\right)^{-n(1+\sigma)}
\left\{
1-\frac{n(1+\sigma)}{2}
\right\}.
\\
(3) 
&&
\frac{d}{dt} \left(\frac{\dot{a}}{{a}}\right)
=
-\frac{n(1+\sigma)H^2}{2}
\left(\frac{a}{a_0}\right)^{-n(1+\sigma)}.
\end{eqnarray*}
\end{lemma}

\begin{proof}
The result (1) follows from $\dot{a}=Ha$ if $\sigma=-1$, and $\dot{a}=Ha(a/a_0)^{-n(1+\sigma)/2}$ if $\sigma\neq-1$.
The results (2) and (3) follow from 
\[
\ddot{a}=\frac{\dot{a}^2}{a}
\left\{
1-\frac{n(1+\sigma)}{2}
\right\}
=
aH^2\left(\frac{a}{a_0}\right)^{-n(1+\sigma)}
\left\{
1-\frac{n(1+\sigma)}{2}
\right\}
\]
and  
\[
\frac{d}{dt} \left(\frac{\dot{a}}{{a}}\right)
=
-\frac{n(1+\sigma)H}{2}
\left(\frac{a}{a_0}\right)^{-n(1+\sigma)/2-1}
\frac{\dot{a}}{a_0}
=
-\frac{n(1+\sigma)H^2}{2}
\left(\frac{a}{a_0}\right)^{-n(1+\sigma)}
\]
by (1), respectively.
\end{proof}

\begin{lemma}
\label{Lem-11}
Let $m\ge0$, $H\in \br$ and $\sigma\in \br$.
Let $T_0$, $a$ and $M$ be given by \eqref{R-Def-T_0}, \eqref{Def-a} and \eqref{M-FLRW}.
Then the following results hold.

(1) 
\ $\displaystyle M^2
=
m^2+\sigma\left(\frac{nH}{2c}\right)^2
\left\{1+\frac{n(1+\sigma)Ht}{2}\right\}^{-2}$.

(2) 
\ $M\dot{M}
=
-c\sigma(1+\sigma)
\left(\frac{nH}{2c}\right)^3
\left\{ 1+\frac{n(1+\sigma)Ht}{2} \right\}^{-3}$.

(3) 
\ $M^2$ satisfies 
\[
\begin{array}{ll}
(i) \ \ M^2=m^2 & \mbox{if }\ H=0, \ \mbox{or}\ \sigma=0,
\\
(ii)\ \ M^2=m^2-\left(\frac{nH}{2c}\right)^2 & \mbox{if }\ H\neq0, \ \sigma=-1,
\\
(iii)\ \ m^2<M^2\le m^2+\sigma\left(\frac{nH}{2c}\right)^2 & \mbox{if }\ H>0, \ \sigma>0,
\\
(iv)\ \ m^2+\sigma\left(\frac{nH}{2c}\right)^2\le M^2<m^2 & \mbox{if }\ (1+\sigma)H>0, \ \sigma<0,
\\
(v)\ \ m^2+\sigma\left(\frac{nH}{2c}\right)^2\le M^2\to \infty \ (t\to T_0) & \mbox{if }\ H<0, \ \sigma>0,
\\
(vi)\ \ m^2+\sigma\left(\frac{nH}{2c}\right)^2\ge M^2\to -\infty \ (t\to T_0) & \mbox{if }\ (1+\sigma)H<0, \ \sigma<0.
\end{array}
\]

(4) 
\ $M^2\ge0$ holds if and only if $H$ and $\sigma$ satisfy one of the following conditions.
\[
\begin{array}{ll}
(i) & H=0, \ \mbox{or}\ \sigma=0,\ m\ge0.
\\
(ii) & H\neq0, \ \sigma>0, \ m\ge0.
\\
(iii) & H\neq0, \ \sigma=-1, \ m\ge \frac{n|H|}{2c}.
\\
(iv) & (1+\sigma)H>0, \ \sigma< 0, \ m\ge \frac{\sqrt{|\sigma| } n|H|}{2c}.
\end{array}
\]

(5)
\ If the case (vi) in (3), and $m>\sqrt{|\sigma|}n|H|/2c$ hold, 
then $0<T_1<T_0<\infty$, $M(T_1)=0$, and $M^2$ is a strictly monotone-decreasing function with $M^2(t)>0$ 
for $0\le t<T_1$,
where $T_1$ is defined by \eqref{Def-T1}.
\end{lemma}

\begin{proof}
The results (1) and (2) follow from the definition of $M$ by \eqref{Def-M}, and (1) and (2) in Lemma \ref{Lem-10}.

(3)
By (1), the results (i) and (ii) are trivial.
By (1), we also have 
\[
\sigma\left(\frac{nH}{2c}\right)^2
\left\{1+\frac{n(1+\sigma)Ht}{2}\right\}^{-2}
\begin{cases}
\searrow 0 & \mbox{if}\ 
H>0,\ \sigma>0,\\
\nearrow 0 & \mbox{if}\ (1+\sigma)H>0, \ \sigma<0, \\
\nearrow \infty & \mbox{if}\ 
H<0,\ \sigma>0,\\
\searrow -\infty & \mbox{if}\ (1+\sigma)H<0, \ \sigma<0, 
\end{cases}
\]
by which we obtain the results from (iii) to (vi).

(4) 
We obtain the required result by (3).

(5) 
We have $0<T_1<T_0<\infty$ by the definitions of $T_0$ and $T_1$ in the case (vi) in (3).
Since $M^2$ is a strictly monotone-decreasing function with $M_0^2>0$ and $M^2(T_1)=0$,
we obtain $M^2>0$ on the interval $[0,T_1)$.
\end{proof}

\begin{lemma}
\label{Lem-11'}
Under the assumption in Lemma \ref{Lem-11}, 
the following results hold.

(1) 
\ $\dot{a}\ge0$ if and only if $H\ge0$.

(2) 
Let $H\ge0$. 
Then $M\dot{M}\le 0$ holds if and only if one of the following conditions from (i) to (iii) holds.

\ \ (i) $H=0$, \ $\sigma\in \br$.

\ \ (ii) $H>0$, \ $\sigma\ge 0$.

\ \ (iii) $H>0$, \ $\sigma\le -1$.

(3) 
The following bounds for $\inf M^2$ hold.

\ \ (i) $\inf_{(0,T_0)} M^2=m^2>0$ if $H=0$,\ $\sigma\in \br$, $m\neq0$.

\ \ (ii) $\inf_{(0,T_0)} M^2=m^2>0$ if $H>0$,\ $\sigma=0$, $m\neq0$.

\ \ (iii) $\inf_{(0,T_0)} M^2=m^2>0$ if $H>0$,\ $\sigma>0$, $m\neq0$. 

\ \ (iv) 
$\inf_{(0,T)} M^2
=
\sigma
\left(
\frac{nH}{2c}
\right)^2\left(\frac{a(T)}{a_0}\right)^{-n(1+\sigma)}>0$ for $0<T<T_0$ 
if $H>0$,\ $\sigma>0$, $m=0$. 

\ \ (v) 
$\inf_{(0,T_0)} M^2=m^2-\left(\frac{nH}{2c}\right)^2>0$ if $H>0$,\ $\sigma=-1$, $m>\frac{nH}{2c}$.

\ \ (vi) 
$\inf_{(0,T_1)} M^2=m^2+\sigma\left(\frac{nH}{2c}\right)^2\left(\frac{a(T)}{a_0}\right)^{-n(1+\sigma)}>0$ for $0<T<T_1$
if $H>0$,\ $\sigma<-1$, $m> \frac{\sqrt{|\sigma|}nH}{2c}$.
\end{lemma}

\begin{proof}
The results (1) and (2) follow from (1) in Lemma \ref{Lem-10} and (2) in Lemma \ref{Lem-11}, respectively.
The result (3) follows from (3) in Lemma \ref{Lem-11}.
\end{proof}


\newsection{Proof of Theorem \ref{Thm-13-Cor-14}}
\label{Sec-Thm-13-Cor-14}
We prove Theorem \ref{Thm-13-Cor-14} in this section.

(1) 
For $T>0$, $R_0>0$, $R_{\mu_0}>0$ and $R_\mu>0$, 
we define the metric space defined by 
\[
X(T,R_0,R_{\mu_0},R_\mu)
:=\{u\in X(T);\ \|u\|_{\dot{X}^\nu(T)}\le R_\nu\ \mbox{for}\ \nu=0,\mu_0,\mu\}
\]
with the metric $d(u,v):=\|u-v\|_{\dot{X}^0(T)}$,
where $X(T)$ and $\|\cdot\|_{\dot{X}^\nu(T)}$ are defined by \eqref{X-Set} and \eqref{Def-XT}.
We regard the solution of the Cauchy problem \eqref{Cauchy} as the fixed point of the operator defined by
\beq
\label{Def-Phi}
\Phi(u)(t)=\Phi(u;u_0,u_1)(t):=K_0(t)u_0+K_1(t) u_1-c^2\int_0^t K_2(t,s) h(u)(s) ds, 
\eeq
where $h(u)(s)=h(u)(s,\cdot):=a^{n/2}(s)f(a^{-n/2}(s)u(s,\cdot))$, 
and $K_0$, $K_1$ and $K_2$ are given by \eqref{Def-K}.
In the following, 
we show that $\Phi$ is a contraction-map on $X(T,R_0,R_{\mu_0},R_\mu)$ for some $T$, $R_0$, $R_{\mu_0}$ and $R_\mu$, 
and we obtain the solution as its fixed point.

Under the assumption of $n$, $\mu_0$, $\mu$, $p$, $q$, $q_\ast$ in the theorem, 
put 
\[ 
\theta:=\frac{(p-1)(n-2\mu_0)}{2p},
\ \ 
\frac{1}{r_\ast}
:=
\frac{1}{2}-\frac{\theta}{n},
\ \ 
\frac{1}{r_\sharp}
:=
\frac{1}{2}-\frac{\mu_0+\theta}{n}.
\]
We have ${1}/{2}=(p-1)/{r_\sharp}+{1}/{r_\ast}$ by the definitions of $r_\ast$ and $r_\sharp$. 
We have $0\le \theta\le 1$ by $\mu_0<n/2$ and the assumption \eqref{Condition-p}. 
We have $0<1/r_\ast\le 1/2$ and $0<1/r_\sharp\le 1/2$ by $\mu_0\ge0$ and $\theta\ge0$.
Thus, we have 
\begin{eqnarray*}
\|f(u)\|_{\dot{H}^{\nu}(\brn)}
&\lesssim& 
\|u\|_{\dot{B}^0_{r_\sharp,2}(\brn)\cap L^{r_\sharp}(\brn) }^{p-1}
\|u\|_{\dot{B}^\nu_{r_\ast,2}(\brn)}
\\
&\lesssim&
\|u\|_{\dot{H}^{\mu_0+\theta}(\brn) }^{p-1}
\|u\|_{\dot{H}^{\nu+\theta}(\brn)}
\\
&\lesssim&
\|u\|_{\dot{H}^{\mu_0}(\brn) }^{(1-\theta)(p-1)}
\|\nabla u\|_{\dot{H}^{\mu_0}(\brn)}^{\theta(p-1)}
\|u\|_{\dot{H}^{\nu}(\brn) }^{1-\theta}
\|\nabla u\|_{\dot{H}^{\nu}(\brn)}^{\theta}
\end{eqnarray*}
for $0\le \nu\le \mu$ by \cite[Lemma 2.2]{Nakamura-Ozawa-1997-RMP}, the Sobolev embedding, and the interpolation.

Since $q_\ast$ defined in \eqref{Def-delta-qAst} satisfies 
\[
0\le \frac{1}{q_\ast} \le 1,
\ \ 
1=\frac{1}{q_\ast}+\frac{\theta p}{q}
\] 
by \eqref{Def-q2},
we have  
\begin{eqnarray}
&&
\|h(u)\|_{L^1((0,T), \dot{H}^{\nu}(\brn)) }
\nonumber\\
&\lesssim&
A_\ast(T)
\|Mu\|_{L^\infty((0,T), \dot{H}^{\mu_0}(\brn)) }^{(1-\theta)(p-1)}
\left\|
\left(\frac{\dot{a}}{a}\right)^{1/q} 
\frac{\nabla u}{a}
\right\|_{L^q((0,T), \dot{H}^{\mu_0}(\brn))}^{\theta(p-1)}
\nonumber\\
&& \cdot \|Mu\|_{L^\infty((0,T), \dot{H}^{\nu}(\brn) )}^{1-\theta}
\left\|
\left(\frac{\dot{a}}{a}\right)^{1/q} 
\frac{\nabla u}{a}
\right\|_{L^q((0,T),\dot{H}^{\nu}(\brn))}^{\theta}
\nonumber 
\\
&\le& A_\ast(T)R_{\mu_0}^{p-1} R_\nu
\label{Proof-Thm-13-1000}
\end{eqnarray}
for $u\in X(T,R_0, R_{\mu_0},R_\mu)$ 
by the H\"older inequality, 
where 
\beq
\label{AAst-A}
A_\ast(T)
:=
\left\|
a^{-n(p-1)/2+\theta p}
\left(
\frac{\dot{a}}{a}
\right)^{-\theta p/q} 
M^{-(1-\theta)p}
\right\|_{L^{q_\ast}((0,T))}\le A(T),
\eeq
where $A(T)$ is defined by \eqref{Def-A}, 
and we have used $(1-\theta)p=\delta$ and $\theta p-n(p-1)/2=-\mu_0(p-1)$ for $\delta$ defined by \eqref{Def-delta-qAst}.

By \eqref{Proof-Thm-13-1000}, \eqref{AAst-A} and the energy estimate (2) in Lemma \ref{Lem-2} putting $h=a^{n/2}f(a^{-n/2} u)$, 
we have 
\[
\|\Phi(u)\|_{\dot{X}^\nu}
\lesssim D_\nu+c\|h(u)\|_{L^1((0,T),\dot{H}^\nu(\brn))}
\lesssim D_\nu+cA(T)R_{\mu_0}^{p-1} R_\nu
\]
for $\nu=0,\mu_0,\mu$, 
where we have put 
\[
D_\nu:= 
c^{-1}\|u_1\|_{\dot{H}^\nu(\brn)}+a_0^{-1}\|\nabla u_0\|_{\dot{H}^\nu(\brn)}+M_0\|u_0\|_{\dot{H}^\nu(\brn)}.
\]
So that, there are some positive constants $C_0$ and $C$ independent of $u_0$, $u$ and $T$ such that 
\beq
\label{Proof-Thm-13-2000}
\|\Phi(u)\|_{\dot{X}^\nu}
\le C_0D_\nu+CcA(T)R_{\mu_0}^{p-1} R_\nu
\le R_\nu
\eeq
for $\nu=0,\mu_0,\mu$, and any $u\in X(T,R_0,R_{\mu_0},R_\mu)$
under the conditions 
\beq
\label{Proof-Thm-13-3000}
2C_0D_\nu\le R_\nu,
\ \ 
2CcA(T)R_{\mu_0}^{p-1}\le 1.
\eeq

Similarly, we estimate the metric as follows.
Starting from the simpler estimate  
\begin{multline*}
\|f(u)-f(v)\|_{L^2(\brn)}
\lesssim
\max_{w=u,v}
\|w\|_{\dot{H}^{\mu_0}(\brn) }^{(1-\theta)(p-1)}
\|\nabla w\|_{\dot{H}^{\mu_0}(\brn)}^{\theta(p-1)}
\\
\cdot 
\|u-v\|_{L^2(\brn) }^{1-\theta}
\|\nabla (u-v)\|_{L^2(\brn)}^{\theta},
\end{multline*}
we have  
\begin{eqnarray*}
&&
\|h(u)-h(v)\|_{L^1((0,T), L^2(\brn)) }
\nonumber\\
&\lesssim&
A(T)
\max_{w=u,v} 
\|Mw\|_{L^\infty((0,T), \dot{H}^{\mu_0}(\brn)) }^{(1-\theta)(p-1)}
\left\|
\left(\frac{\dot{a}}{a}\right)^{1/q} 
\frac{\nabla w}{a}
\right\|_{L^q((0,T), \dot{H}^{\mu_0}(\brn))}^{\theta(p-1)}
\nonumber\\
&& \cdot \|M(u-v)\|_{L^\infty((0,T), L^2(\brn) )}^{1-\theta}
\left\|
\left(\frac{\dot{a}}{a}\right)^{1/q} 
\frac{\nabla (u-v)}{a}
\right\|_{L^q((0,T),L^2(\brn))}^{\theta}
\nonumber 
\\
&\le& A(T)R_{\mu_0}^{p-1} d(u,v)
\end{eqnarray*}
and 
\beq
\label{Proof-Thm-13-4000}
\begin{array}{lll}
d(\Phi(u;u_0,u_1),\Phi(v;v_0,v_1))
&\lesssim& \widetilde{D}+c\|h(u)-h(v)\|_{L^1((0,T),L^2(\brn))} \\
&\lesssim& \widetilde{D}+cA(T)R_{\mu_0}^{p-1} d(u,v)
\end{array}
\eeq
for any $u, v\in X(T,R_0,R_{\mu_0},R_\mu)$, 
where we have put 
\[
\widetilde{D}:= 
c^{-1}\|u_1-v_1\|_{L^2(\brn)}+a_0^{-1}\|\nabla (u_0-v_0)\|_{L^2(\brn)}+M_0\|u_0-v_0\|_{L^2(\brn)}.
\]
So that, we obtain 
\beq
\label{Proof-Thm-13-5000}
d(\Phi(u;u_0,u_1),\Phi(v;v_0,v_1))
\le C_0\widetilde{D}+CcA(T)R_{\mu_0}^{p-1} d(u,v)
\le C_0\widetilde{D}+\frac{1}{2} d(u,v)
\eeq
for any $u,v\in X(T,R_0,R_{\mu_0},R_\mu)$
under the conditions \eqref{Proof-Thm-13-3000},
where $C_0$ and $C$ are newly taken if necessary.

When $v$ has the same initial data of $u$, i.e., $u_0=v_0$ and $u_1=v_1$, we have $\widetilde{D}=0$.
By \eqref{Proof-Thm-13-2000} and \eqref{Proof-Thm-13-5000}, the operator $\Phi$ is a contraction-mapping on $X(T,R_0,R_{\mu_0},R_\mu)$ under the conditions \eqref{Proof-Thm-13-3000}, 
and $\Phi$ has a unique fixed point $u$ in it.
We can show $u\in C([0,T),H^{\mu+1}(\brn))$ $\cap C^1([0,T),H^\mu(\brn))$, 
the uniqueness of $u$ in $X(T)$, and the result (2) by standard arguments 
(see, e.x., \cite{Nakamura-2021-JDE}).

(3) 
The conditions \eqref{Proof-Thm-13-3000} are satisfied by 
\beq
\label{Proof-Thm-13-6000}
2Cc(2C_0D_{\mu_0})^{p-1} A(T)\le 1
\eeq
by taking $R_{\mu_0}=2C_0 D_{\mu_0}$ 
for any $T$ with $0<T\le T_1\le T_0$.
So that, if $T=T_1$, then the solution exists on the interval $[0,T_1)$.
Moreover, if $T=T_0$, then the solution $u$ obtained by (1) exists globally 
since $T_0$ is the end of the spacetime, 
and it satisfies 
\[
u(t)=
u_+(t)
+
c^2\int_t^{T_0} K_2(t,s) h(u)(s) ds,
\]
where 
\begin{eqnarray*}
&&
v_0:=
u_0+c^2\int_0^{T_0} K_1(s) h(u)(s) ds,
\ \ 
v_1:=
u_1-c^2\int_0^{T_0} K_0(s) h(u)(s) ds,
\\
&& 
u_+(t):=K_0(t)v_0+K_1(t)v_1
\end{eqnarray*}
for $0\le t<T_0$ 
by $\Phi(u)=u$ in \eqref{Def-Phi}.
Let $\eta$ and $N_1,N_2,N_3,N_4$ be defined by \eqref{Def-Eta} and \eqref{Def-N}.
We note that $N_1,N_2,N_3,N_4$ are finite numbers by the assumption 
$\inf_{0<t<T_0} M(t)>0$.
The assumption $\alpha\ge0$ and $\partial_t\alpha\le 0$ in (2) in Lemma \ref{Lem-4} for $\alpha$ defined by \eqref{Def-Alpha} is satisfied as $\alpha>0$ by its definition and (2) in Lemma \ref{Lem-5}. 
We have 
\[
\|v_0\|_{H^{\nu}(\brn)}
\le
\|u_0\|_{H^\nu(\brn)}
+
cN_4\|h(u)\|_{L^1((0,T_0),H^\nu(\brn))},
\]
\[
\|v_1\|_{H^{\nu-1}(\brn)}
\le
\|u_1\|_{H^{\nu-1}(\brn)}
+
c^2N_1\|h(u)\|_{L^1((0,T_0),H^\nu(\brn))},
\]
\[
\|u_+\|_{H^{\nu-1}(\brn)}
\le
N_1\|v_0\|_{H^{\nu}(\brn)}
+
\frac{N_4}{c}\|v_1\|_{H^{\nu-1}(\brn)}
\]
hold for $0\le \nu\le \mu$ by (1) and (3) in Lemma \ref{Lem-7},
and each right hand side is finite by \eqref{Proof-Thm-13-1000}. 
Thus, putting $\nu=\mu-1$ or $\nu=\mu$, we have 
\[
\|u(t)-u_+(t)\|_{H^{\mu-1}(\brn)}
\lesssim
cN_1N_4 \int_t^{T_0} \|h(u)(s)\|_{H^\mu(\brn)} ds \to 0,
\]
and 
\[
\eta(t)^{-1}\|u(t)-u_+(t)\|_{H^{\mu}(\brn)}
\lesssim
c\max\{N_1N_3, N_4\}  \int_t^{T_0} \|h(u)(s)\|_{H^\mu(\brn)} ds \to 0
\]
as $t$ tends to $T_0$ by (8) in Lemma \ref{Lem-7},
where we have used the monotonicity of $\eta$ by (7) in Lemma \ref{Lem-5}.
Similarly, since we have
\[
\partial_t(u-u_+)(t)=c^2\int_t^{T_0} \partial_tK_2(t,s) h(u)(s) ds
\]
by $K_2(t,t)=0$, 
we have 
\[
\|\partial_t(u-v)(t)\|_{H^{\mu-1}(\brn)}
\lesssim
c^2\max\{N_1,N_2N_4\} \int_t^{T_0} \|h(u)(s)\|_{H^\mu(\brn)} ds \to 0
\]
and 
\[
\eta(t)^{-1}\|\partial_t(u-v)(t)\|_{H^{\mu}(\brn)}
\lesssim
c^2\max\{1, N_2N_3\}  \int_t^{T_0} \|h(u)(s)\|_{H^\mu(\brn)} ds \to 0
\]
as $t$ tends to $T_0$ by (9) in Lemma \ref{Lem-7}.

(4) 
We prove the existence of global solutions for large data.
By the energy estimate \eqref{Lem-2-5} with \eqref{f-V}, 
we have 
\[
c^{-1}\|\partial_t u(t)\|_{L^2(\brn)}+a^{-1}\|\nabla u(t)\|_{L^2(\brn)}+M(t)\|u(t)\|_{L^2(\brn)}\le 3\sqrt{\widetilde{E}(u)(0)}
\]
by $\dot{a}\ge0$, $\dot{M}\le 0$ and $\lambda\ge0$, 
where $\widetilde{E}(u)$ is defined by \eqref{Lem-2-5}.
Taking $q=\infty$, 
we have $q_\ast=1$ by \eqref{Def-delta-qAst}.
Let $u$ be the solution of \eqref{Cauchy} by (1).
Now, assume that there exists $T_\ast$ such that $0<T_\ast<T_0$ with 
$\limsup_{t\nearrow T_\ast} \|u(t)\|_{H^{\mu+1}(\brn)}=\infty$, 
then we can take positive real numbers $t_0$, $\varepsilon$ and $T$ such that 
$0<t_0<T_\ast<t_0+\varepsilon<T<T_0$ and 
\[
\widetilde{A}(T):=M(T)^{-\delta}  
\left\|
{a}^{-\mu_0(p-1)}
\left(
\frac{\dot{a}}{a}
\right)^{1/q_\ast-1}
\right\|_{L^{q_\ast}((t_0,t_0+\varepsilon))}
\]
and 
\[
2Cc \widetilde{A}(T)\left(6C_0\sqrt{\widetilde{E}(u)(0)}\right)^{p-1}\le 1,
\]
similarly to \eqref{Proof-Thm-13-3000} and \eqref{Def-A}.
So that, we can obtain the prolonged solution $u\in C([t_0,t_0+\varepsilon), H^{\mu+1}(\brn))$ starting from $t_0$ by the same argument in the proof of (1),
which contradicts to the definition of $T_\ast$.
So that, we have 
\[
\limsup_{0<t<T_0} \|u(t)\|_{H^{\mu+1}(\brn)}<\infty,
\]
which shows that 
$u$ is a global solution since $u$ exists on the small interval starting from any $t\in [0,T_0)$ by the same argument in the proof of (1).

\newsection{Proof of Corollary \ref{Thm-16-Cor-17}}
\label{Sec-Thm-16-Cor-17}
To prove Corollary \ref{Thm-16-Cor-17}, 
we prepare the following lemma which estimates the existence-time of the solution 
under \eqref{Def-a} and \eqref{M-FLRW}. 

\begin{lemma}
\label{Lem-15}
Let $n\ge1$, and let $a$ be given by \eqref{Def-a}.
Let $H$ and $\sigma$ satisfy one of the conditions from (i) to (iii) in (2) in Lemma \ref{Lem-11'}.
Put 
\beq
\label{Def-B}
B(T):=
\left\|
\left(
\frac{a}{a_0}
\right)^{-\mu_0(p-1)}
\left(
\frac{\dot{a}}{a}
\right)^{1/q_\ast-1}
\right\|_{L^{q_\ast}((0,T))}
\eeq
for $0<T\le T_0$.
Let $q$ and $q_\ast$ be those in Theorem \ref{Thm-13-Cor-14}.
Let $\gamma$, $p_1(\mu_0)$ and $p_2(\mu_0)$ be given by \eqref{Def-G-Gamma} and \eqref{Def-p1-p2}.
Let $m>\sqrt{|\sigma|}nH/2c$ 
if $H>0$ and $\sigma<-1$.
Then the following results hold, 
where $B_1$, $B_2$ and $B_3$ are defined by 
\eqref{Def-J0-J1} and \eqref{Def-J2-J3}.

(1) 
$B(T)=T$ 
if $H=0$, $\sigma\in \br$, $q=\infty$, $1\le p<\infty$.

(2) 
$\gamma>1$ and 
\[
B(T)=
B_1
\left[
1-
\left\{
1+\frac{n(1+\sigma)HT}{2}
\right\}^{1-\gamma}
\right]^{1/q_\ast}\le B_1
\]
if $H>0$, $\sigma\ge0$, $\mu_0>0$, $p_1(\mu_0)<p<\infty$, $q_\ast<\infty$.

(3) 
$\gamma<1$ and 
\[
B(T)=
B_1
\left|
\left\{
1+\frac{n(1+\sigma)HT}{2}
\right\}^{1-\gamma}
-1
\right|^{1/q_\ast}
\]
if $q_\ast<\infty$ and one of the following conditions from (i) to (iv) holds. 
Moreover, if $\sigma<-1$ under one of (i), (ii), (iii) and (iv), 
then $B(T)\le B_1$ holds.

(i) $H>0$, $\sigma<-1$, $\mu_0=0$, $1\le p<\infty$.

(ii) $H>0$, $\sigma\ge 0$, $\mu_0=0$, $1\le p<\infty$.

(iii) $H>0$, $\sigma<-1$, $\mu_0>0$, $1\le p<\infty$.

(iv) $H>0$, $\sigma\ge 0$, $\mu_0>0$, $1\le p<p_1(\mu_0)$.

(4) 
$\gamma=1$ and 
\[
B(T)=
B_2
\left[
\log 
\left\{
1+\frac{n(1+\sigma)HT}{2}
\right\}
\right]^{1/q_\ast}
\]
if $H>0$, $\sigma\ge0$, $\mu_0>0$, $p=p_1(\mu_0)$, $q_\ast<\infty$.

(5) 
\[
B(T)=
B_3
\left\{
1-e^{-\mu_0(p-1)HTq_\ast}
\right\}^{1/q_\ast}
\]
if $H>0$, $\sigma=-1$, $\mu_0>0$, $1<p<\infty$, $q_\ast<\infty$.

(6) 
$B(T)=(2HT)^{1/q_\ast}/2H$ 
if one of the following conditions holds.
(i) $H>0$, $\sigma=-1$, $\mu_0=0$, $1\le p<\infty$, $q_\ast<\infty$.
(ii) $H>0$, $\sigma=-1$, $\mu_0>0$, $p=1$, $q_\ast<\infty$.

(7) 
$B(T)=(2H)^{-1}(a(T)/a_0)^{-\mu_0(p-1)+n(1+\sigma)/2}$ 
if one of the following conditions holds.
(i) $H>0$, $\sigma\ge0$, $\mu_0=0$, $p_2(\mu_0)\le p<\infty$, $q_\ast=\infty$.
(ii) $H>0$, $\sigma\ge0$, $\mu_0>0$, $p_2(\mu_0)\le p<p_1(\mu_0)$, $q_\ast=\infty$.

(8) 
$B(T)=(2H)^{-1}$ 
if one of the following conditions holds.
(i) $H>0$, $\sigma\ge0$, $\mu_0>0$, $\max\{p_1(\mu_0),p_2(\mu_0)\}\le p<\infty$, $q_\ast=\infty$.
(ii) $H>0$, $\sigma\le-1$, $\mu_0\ge0$, $p_2(\mu_0)\le p<\infty$, $q_\ast=\infty$.
\end{lemma}

\begin{proof}
We have the result (1), i.e., $B(T)=T$, since $a$ is a constant if $H=0$, and $q_\ast=1$ by $q=\infty$. 
To prove the results (2), (3) and (4), let $H>0$, $\sigma\in \br$ and $1\le q_\ast\le \infty$.
We have 
\beq
\label{Proof-Lem-15-1000}
B(T)=
(2H)^{1/q_\ast-1}
\left\|
\left(
\frac{a}{a_0}
\right)^{-\mu_0(p-1)-n(1+\sigma)(1/q_\ast-1)/2}
\right\|_{L^{q_\ast}((0,T))} 
\eeq
by Lemma \ref{Lem-10}.
If $\sigma\neq-1$ and $q_\ast<\infty$, 
then we have 
\begin{eqnarray*}
B(T)
&=&
(2H)^{1/q_\ast-1}
\left[
\int_0^T
\left\{
1+\frac{n(1+\sigma)Ht}{2}
\right\}^{-\gamma} dt 
\right]^{1/q_\ast}
\\
&=& 
\frac{1}{2H}
\left\{
\frac{4}{n(1+\sigma)}
\right\}^{1/q_\ast}
\cdot
\begin{cases}
\left(
\frac{1}{1-\gamma}
\left[
\left\{
1+\frac{n(1+\sigma)HT}{2}
\right\}^{1-\gamma}
-1
\right]
\right)^{1/q_\ast}
& 
\mbox{if}\ \gamma\neq1, 
\\
\left[
\log 
\left\{
1+\frac{n(1+\sigma)HT}{2}
\right\}
\right]^{1/q_\ast}
&
\mbox{if}\ \gamma=1, 
\end{cases}
\\
&\le&
B_1 \ \ 
\mbox{if}\ 
\begin{cases}
H>0,\ \sigma>-1, \gamma>1,
\\
\mbox{or} \ \ 
H>0,\ \sigma<-1, \gamma<1.
\end{cases}
\end{eqnarray*}
Under $q_\ast<\infty$, we note the following equivalence.
\begin{eqnarray*}
\gamma>1 &\Leftrightarrow& \mu_0>0, \ \sigma>-1, \ p_1(\mu_0)<p<\infty. 
\\
\gamma<1 &\Leftrightarrow& 
\begin{cases}
\mu_0=0, \ \sigma\neq-1, \ 1\le p<\infty,
\\
\mbox{or}
\ \ 
\mu_0>0, \ \sigma<-1, \ 1\le p<\infty,
\\
\mbox{or}
\ \ 
\mu_0>0, \ \sigma>-1, \ 1\le p<p_1(\mu_0).
\end{cases}
\\
\gamma=1 &\Leftrightarrow& \mu_0>0, \ \sigma>-1, \ p=p_1(\mu_0),
\end{eqnarray*}
where we note that $\sigma>-1$ can be replaced with $\sigma\ge0$ by $H>0$ and the assumption (i), (ii) or (iii) in (2) in Lemma \ref{Lem-11'}.
The results (2), (3) and (4) corresponds to the cases $\gamma>1$, $\gamma<1$ and $\gamma=1$, respectively.

To prove the results (5) and (6), let $H>0$, $\sigma=-1$ and $q_\ast<\infty$.
We have 
\begin{eqnarray*}
B(T)
&=& 
\begin{cases}
B_3
\left\{
1-e^{-\mu_0(p-1)HTq_\ast}
\right\}^{1/q_\ast}
& 
\mbox{if}\ \mu_0(p-1)>0, 
\\
\frac{(2HT)^{1/q_\ast}}{2H}
& 
\mbox{if}\ \mu_0(p-1)=0
\end{cases}
\end{eqnarray*}
by \eqref{Proof-Lem-15-1000}.
The results (5) and (6) correspond to the cases $\mu_0(p-1)>0$ and $\mu_0(p-1)=0$, respectively.

To prove the results (7) and (8), let $H>0$, $\sigma\in \br$ and $q_\ast=\infty$, 
and put $\omega:=-\mu_0(p-1)+n(1+\sigma)/2$.
We have $q=(p-1)(n-2\mu_0)/2$ by the definition \eqref{Def-delta-qAst} of $q_\ast$.
We need $p_2(\mu_0)\le p<\infty$ to have $2\le q<\infty$.
We have 
\begin{eqnarray*}
B(T)
&=& 
\begin{cases}
(2H)^{-1}
& 
\mbox{if}\ 
\begin{cases}
\sigma\le -1, \ \omega\in \br,
\\
\mbox{or}
\ \ 
\sigma\ge 0,\ \omega\le 0, 
\end{cases}
\\
(2H)^{-1}
\left\{
\frac{a(T)}{a_0}
\right\}^{\omega}
& 
\mbox{if}\ \sigma\ge 0,\ \omega> 0 
\end{cases}
\end{eqnarray*}
by \eqref{Proof-Lem-15-1000},
where we note that $a$ is a non-decreasing function.
Under $\sigma\ge0$, 
we note the following equivalence.
\begin{eqnarray*}
\omega\le0 &\Leftrightarrow& \mu_0>0, \ p_1(\mu_0)\le p<\infty. 
\\
\omega>0 &\Leftrightarrow& 
\begin{cases}
\mu_0=0, \ 1\le p<\infty,
\\
\mbox{or}
\ \ 
\mu_0>0, \ 1\le p<p_1(\mu_0).
\end{cases}
\end{eqnarray*}
The result (7) corresponds to the case $\sigma\ge0$ with $\omega>0$,
and the result (8) corresponds to the cases $\sigma\le -1$ with $\omega\in \br$, or $\sigma\ge0$ with $\omega\le0$.
\end{proof}

We prove Corollary \ref{Thm-16-Cor-17}.
We recall the condition \eqref{Proof-Thm-13-3000} under which we have obtained the solution of the Cauchy problem \eqref{Cauchy}.
This condition is satisfied by \eqref{Proof-Thm-13-6000} 
by taking $R_{\mu_0}=2C_0D_{\mu_0}$,
which is rewritten as 
\beq
\label{Proof-Thm-16-2000}
B(T)\le GM(T)^\delta,
\eeq
by the definitions of $M$, $\delta$, $G$ and $B$ 
by \eqref{Def-M}, \eqref{Def-delta-qAst}, \eqref{Def-G-Gamma} and \eqref{Def-B}, respectively.

(1) 
We prove the existence of local solutions.
It suffices to show that \eqref{Proof-Thm-16-2000} holds under each assumption from (i) to (xiii) in the statement.

(i) 
Under the assumption $H=0$, we have $B=T$, and $M=m$ by (1) in Lemma \ref{Lem-15}, and  (3) in Lemma \ref{Lem-11}.
Thus, the condition \eqref{Proof-Thm-16-2000} holds if $T$ satisfies  
$T\le G m^\delta$ as required.

(ii)
Under the assumption, 
we have $\gamma>1$, and the condition \eqref{Proof-Thm-16-2000} is satisfied by   
\beq
\label{Proof-Thm-16-(2)}
B_1
\left[
1-
\left\{
1+\frac{n(1+\sigma)HT}{2}
\right\}^{1-\gamma}
\right]^{1/q_\ast}
\le Gm^\delta
\eeq
by (2) in Lemma \ref{Lem-15},  
and $0<m\le M$ by (3) in Lemma \ref{Lem-11}. 
This condition \eqref{Proof-Thm-16-(2)} holds for $T$ with 
\beq
\label{Proof-Thm-16-(2)-1000}
T
\le 
\begin{cases}
T_1(=T_0) 
& \mbox{if}\ B_1\le Gm^\delta, 
\\
\frac{2}{n(1+\sigma)H}
\left[
\left\{
1-\left(
\frac{Gm^\delta}{B_1}
\right)^{q_\ast}
\right\}^{-1/(\gamma-1)}-1
\right]
& \mbox{if}\ B_1> Gm^\delta,
\end{cases}
\eeq
where we note that $B_1\le Gm^\delta$ is equivalent to 
\beq
\label{Proof-Thm-16-(2)-2500}
D_{\mu_0}\le \frac{a_0^{\mu_0}}{2C_0}
\left(
\frac{m^\delta}{2CcB_1}
\right)^{1/(p-1)}.
\eeq
So that, we obtain the required result.

(iii) 
Under the assumption, 
we have $\gamma>1$, and the condition \eqref{Proof-Thm-16-2000} is written as 
\[ 
B_1
\left[
1-\left\{
1+\frac{n(1+\sigma)HT}{2}
\right\}^{1-\gamma}
\right]^{1/q_\ast}
\le G
\left[
\sigma
\left(\frac{nH}{2c}\right)^2
\left\{
1+\frac{n(1+\sigma)HT}{2}
\right\}^{-2}
\right]^{\delta/2}
\]
by (2) in Lemma \ref{Lem-15},
which is the required result.

(iv), (v)  
Under the assumption, 
we have $\gamma<1$, and the condition \eqref{Proof-Thm-16-2000} is written as 
\beq
\label{Proof-Thm-16-(4)}
B_1
\left[
\left\{
1+\frac{n(1+\sigma)HT}{2}
\right\}^{1-\gamma}
-1
\right]^{1/q_\ast}
\le GM(T)^\delta
\eeq
by (3) in Lemma \ref{Lem-15},
which is satisfied if $T$ satisfies 
\[
0<T\le 
\frac{2}{n(1+\sigma)H}
\left[
\left\{
1+\left(
\frac{Gm^\delta}{B_1}
\right)^{q_\ast}
\right\}^{1/(1-\gamma)}-1
\right],
\]
where we have used $0<m\le M(T)$ by (3) in Lemma \ref{Lem-11}.

(vi) 
Under the assumption, 
we have $\gamma<1$, and the condition \eqref{Proof-Thm-16-2000} is written as 
\[ 
B_1
\left[
\left\{
1+\frac{n(1+\sigma)HT}{2}
\right\}^{1-\gamma}
-1
\right]^{1/q_\ast}
\le G
\left[
\sigma
\left(\frac{nH}{2c}\right)^2
\left\{
1+\frac{n(1+\sigma)HT}{2}
\right\}^{-2}
\right]^{\delta/2}
\]
by (3) in Lemma \ref{Lem-15},
which is the required result.

(vii) 
Under the assumption, 
we have $\gamma<1$, and 
the condition \eqref{Proof-Thm-16-2000} is written as 
\[ 
B_1
\left[
1-
\left\{
1+\frac{n(1+\sigma)HT}{2}
\right\}^{1-\gamma}
\right]^{1/q_\ast}
\le G
\left[
m^2+
\sigma
\left(\frac{nH}{2c}\right)^2
\left\{
1+\frac{n(1+\sigma)HT}{2}
\right\}^{-2}
\right]^{\delta/2}
\]
by (3) in Lemma \ref{Lem-15},
which yields the required result.

(viii) 
Under the assumption, 
we have $\gamma=1$, 
and the condition \eqref{Proof-Thm-16-2000} is rewritten as 
\beq
\label{Proof-Thm-16-(7)}
B_2
\left[ \log
\left\{
1+\frac{n(1+\sigma)HT}{2}
\right\}
\right]^{1/q_\ast}
\le GM(T)^\delta
\eeq
by (4) in Lemma \ref{Lem-15}.
We rewrite this as 
\[
T\le \frac{2}{n(1+\sigma)H}
\left\{
e^{(GM(T)^\delta/B_2)^{q_\ast} }-1
\right\},
\]
which is satisfied by the required result by $0<m\le M(T)$ due to (3) in Lemma \ref{Lem-11}.

(ix)
Similarly to (viii), we have $\gamma=1$ and \eqref{Proof-Thm-16-(7)}.
We obtain the required result by (1) in Lemma \ref{Lem-11} by $m=0$.

(x) 
Under the assumption, 
the condition \eqref{Proof-Thm-16-2000} is written as 
\[
B_3\left\{
1-e^{-\mu_0(p-1)HTq_\ast} 
\right\}^{1/q_\ast}
\le GM(T)^\delta
\]
by (5) in Lemma \ref{Lem-15},
which is satisfied if $T$ satisfies 
\beq
\label{Proof-Thm-16-ix-1000}
T
\le 
\begin{cases}
T_1(=T_0=\infty) 
& \mbox{if}\ B_3\le G
\left\{ 
m^2-\left(\frac{nH}{2c}\right)^2
\right\}^{\delta/2}, 
\\
-\frac{1}{\mu_0(p-1)Hq_\ast} 
\log
\left(
1-
\left[
\frac{G}{B_3}
\left\{
m^2-\left(\frac{nH}{2c}\right)^2\right
\}^{\delta/2} 
\right]^{q_\ast}
\right)
& \mbox{if}\ B_3> G
\left\{ 
m^2-\left(\frac{nH}{2c}\right)^2
\right\}^{\delta/2}, 
\end{cases}
\eeq
which is the required result,
where we note that 
$B_3\le G\left\{m^2-({nH}/{2c})^2\right\}^{\delta/2}$ 
is equivalent to
\beq
\label{Proof-Thm-16-ix-2000}
D_{\mu_0}
\le 
\frac{a_0^{\mu_0}}{2C_0}
\left(
\frac{ \left\{m^2-(nH/2c)^2 \right\}^{\delta/2} }{2CcB_3}
\right)^{1/(p-1)}.
\eeq
So that, we obtain the required result.

(xi) 
Under the assumption, 
we have $B=(2HT)^{1/q_\ast}/2H$ by (6) in Lemma \ref{Lem-15}, 
and the condition \eqref{Proof-Thm-16-2000} is easily rewritten as 
the required condition for $T$.

(xii) 
Under the assumption, 
the condition \eqref{Proof-Thm-16-2000} is rewritten as \eqref{Thm-16-11-2} 
by (1) in Lemma \ref{Lem-11} and (7) in Lemma \ref{Lem-15}, 
which requires \eqref{Thm-16-11-1} for $T>0$.

(xiii) 
Under the assumption, 
the condition \eqref{Proof-Thm-16-2000} is rewritten as 
\beq
\label{Proof-Thm-16-(13)}
1\le 2HGM(T)^{\delta}
\eeq
by (7) in Lemma \ref{Lem-15}.
When $\delta=0$, i.e., $p=p(\mu_0)$, 
the inequality \eqref{Proof-Thm-16-(13)} is rewritten as 
\[
D_{\mu_0}
\le \frac{a_0^{\mu_0}}{2C_0} \left(\frac{H}{Cc}\right)^{1/(p-1)},
\]
where $T$ can be taken arbitrarily under $0<T<T_1$ 
and $m>\sqrt{|\sigma|}nH/2c$ 
by (5) in Lemma \ref{Lem-11}.
When $\delta>0$, i.e., $p<p(\mu_0)$, 
the inequality \eqref{Proof-Thm-16-(13)} holds if and only if 
\beq
\label{Proof-Thm-16-(13)-1000}
m>(2HG)^{-1/\delta},
\ \ 
T\le 
-\frac{2}{n(1+\sigma)H}
\left[
1-\sqrt{\frac{|\sigma|}{m^2-(2HG)^{-2/\delta} } }
\frac{nH}{2c}
\right]
\eeq
by (1) in Lemma \ref{Lem-11} with  
\[
D_{\mu_0}
<
\frac{a_0^{\mu_0}}{2C_0}
\left[
\frac{H}{Cc}
\left\{
m^2+\sigma
\left(
\frac{nH}{2c}
\right)^2
\right\}^{\delta/2}
\right]^{1/(p-1)},
\]
where the condition 
$nH\sqrt{{|\sigma|}/\{m^2-(2HG)^{-2/\delta} \} }
/{2c}<1$ requires 
$m^2-(2HG)^{-2/\delta}>|\sigma|(nH/2c)^2$ 
which yields the first inequality in \eqref{Proof-Thm-16-(13)-1000}.

(2) 
We prove the existence of global solutions for small data.
It suffices to show that the inequality \eqref{Proof-Thm-16-2000} holds for $T=T_0$ 
under each assumption from (i) to (iv) in the statement.

(i) 
This case corresponds to the case (ii) in (1).
We can take $T=T_0$ under the condition 
\eqref{Proof-Thm-16-(2)-2500} by \eqref{Proof-Thm-16-(2)-1000}.

(ii) 
This case corresponds to the case (x) in (1).
We can take $T=T_0$ under the condition \eqref{Proof-Thm-16-ix-2000} by \eqref{Proof-Thm-16-ix-1000}.

(iii), (iv) 
The condition \eqref{Proof-Thm-16-2000} is rewritten as 
\[
D_{\mu_0}
\le 
\frac{a_0^{\mu_0}}{2C_0}
\left\{
\frac{HM(T)^\delta}{Cc}
\right\}^{1/(p-1)}
\]
by (8) in Lemma \ref{Lem-15}.
This inequality holds for $T=T_0$ 
using $0<m\le M$ if $\sigma\ge0$, 
and $\sqrt{m^2-(nH/2c)^2}\le M(T)$ if $\sigma=-1$ by (3) in Lemma \ref{Lem-11}.
The results (iii) and (iv) follow from the cases $\sigma\ge0$ and $\sigma=-1$,
respectively.

(3) 
The results directly follow from (4) in Theorem \ref{Thm-13-Cor-14}.

\newsection{Proof of Theorem \ref{Thm-19}}
\label{Sec-Thm-19}
We put $g:=a^2\|u\|_2^2$.
Since we have 
\[
a^2\frac{d^2}{dt^2}\|u\|^2_2=2a^2\|\partial_t u\|_2^2-2c^2\|\nabla u\|_2^2-2c^2M^2 g-(p+1)c^2a^{2-n(p-1)/2}\int_\brn V(u) dx
\]
by the energy estimates \eqref{f-V} and (5) in Lemma \ref{Lem-2}, 
we obtain
\begin{eqnarray}
\label{Proof-Thm-19-1000}
\ddot{g}
&=&
-I\cdot 
g
+4a^{-1}\dot{a}\dot{g}
+2a^2\|\partial_t u\|_2^2
-2c^2\|\nabla u\|_2^2
\nonumber\\
&&
-(p+1)c^2a^{2-n(p-1)/2}
\int_\brn V(u) dx
\end{eqnarray}
and 
\[
a^2\widetilde{F}^0+\widetilde{F}^1=a_0^2 \widetilde{F}^0_0,
\]
where we have put 
\[
I:=
-2\frac{d}{dt}
\left(
\frac{\dot{a}}{a}
\right)
+4\left(\frac{\dot{a}}{a}\right)^2
+2c^2M^2,
\]
\[
\widetilde{F}^0(t):=\int_\brn \widetilde{e}^0(t,x) dx,
\ \ 
\widetilde{F}^0_0:= \widetilde{F}^0(0)
\]
and 
\beq
\label{Proof-Thm-19-3000}
\widetilde{F}^1(t)
:=
\int_0^t\int_\brn 
-2a(s)\dot{a}(s)\widetilde{e}^0(s,x)+a^2(s)\widetilde{e}^{n+1}(s,x) dx ds
\eeq
for $\widetilde{e}^0$ and $\widetilde{e}^{n+1}$ in Lemma \ref{Lem-2}.
We note 
\begin{eqnarray}
\label{Proof-Thm-19-2000}
\kappa c^2 a^2 \widetilde{F}^0
&=&
\kappa a^2 \|\partial_tu\|_2^2+\kappa c^2\|\nabla u\|_2^2
\nonumber\\
&& 
+\kappa c^2 a^2 M^2 \|u\|_2^2+\kappa c^2a^{2-n(p-1)/2}\int_\brn V(u) dx
\end{eqnarray}
for $\kappa\in \br$ by the definition of $\widetilde{e}^0$.

By \eqref{Proof-Thm-19-1000} and \eqref{Proof-Thm-19-2000},
we have 
\begin{eqnarray}
\label{Proof-Thm-19-2500}
-\kappa c^2 a_0^2 \widetilde{F}^0_0
&=&
\ddot{g}+\two g-4a^{-1}\dot{a}\dot{g}-(\kappa+2)a^2\|\partial_tu\|_2^2
+R
\nonumber\\
&\le& 
\ddot{g}+\two g-4a^{-1}\dot{a}\dot{g}-(\kappa+2)a^2\|\partial_tu\|_2^2
\end{eqnarray}
if $R\le 0$, where we have put 
$\two:=I-\kappa c^2 M^2$ and 
\[
R:=-(\kappa-2)c^2\|\nabla u\|_2^2
+
(p+1-\kappa)c^2a^{2-n(p-1)/2} \int_\brn V(u) dx
-
\kappa c^2 \widetilde{F}^1.
\]
Since we have 
\begin{multline*}
-2a\dot{a} \widetilde{e}^0+a^2\widetilde{e}^{n+1}
=
-2c^{-2}a\dot{a} |\partial_t u|^2-2(a\dot{a}M^2+a^2M\dot{M})|u|^2
\\
+\left\{ \frac{n(p-1)}{2}-2\right\} a^{1-n(p-1)/2}  \dot{a} V(u),
\end{multline*}
we obtain $\widetilde{F}^1 \ge 0$ 
by \eqref{Proof-Thm-19-3000} 
if $\dot{a}\le0$, $\dot{M}\le 0$ and $\lambda\{n(p-1)-4\}\le 0$.
Provided $\widetilde{F}^1\ge 0$, we also obtain 
$R\le 0$ if $\kappa\ge2$ and $\lambda(p+1-\kappa)\le 0$.

We have
\beq
\label{Proof-Thm-19-4000}
-\kappa c^2 a_0^2 \widetilde{F}^0_0g
\le 
g\ddot{g}+\two g^2-4a^{-1}\dot{a}g\dot{g}-(\kappa+2)a^2\|\partial_tu\|_2^2g
\eeq
multiplying $g$ to the both sides of the last inequality in \eqref{Proof-Thm-19-2500}.
Since we have 
\beq
\label{Proof-Thm-19-4500}
\dot{g}=2a^2\Re\int \overline{\partial_t u} u dx+2a^{-1}\dot{a} g
\eeq
by direct calculation, 
we have 
\beq
\label{Proof-Thm-19-5000}
\left|
\dot{g}-2a^{-1}\dot{a} g
\right|^2
\le 4 a^2\|\partial_t u\|_2^2 g.
\eeq
By \eqref{Proof-Thm-19-4000} and \eqref{Proof-Thm-19-5000},
we obtain 
\beq
\label{Proof-Thm-19-6000}
-\kappa c^2 a_0^2 \widetilde{E}^0_0g
\le
g\ddot{g}+\three g^2+(\kappa-2)a^{-1}\dot{a} g\dot{g}-\frac{(\kappa+2)(\dot{g})^2}{4}, 
\eeq
where we have put 
\[
\three:=\two-(\kappa+2)
\left(
\frac{\dot{a}}{a}
\right)^2.
\]
So that, we obtain 
\beq
\label{Proof-Thm-19-7000}
0
<
g\ddot{g}+(\kappa-2)a^{-1}\dot{a} g\dot{g}-\frac{(\kappa+2)(\dot{g})^2}{4} 
\eeq
by $\widetilde{F}^0_0<0$ and $\three\le 0$  
which are the assumptions \eqref{Thm-19-4000} and \eqref{Thm-19-3000}, respectively, 
since $\three$ is rewritten as 
\[
\three=
-(\kappa-2)
\left\{
c^2M^2
+
\left(
\frac{ \dot{a} }{a}
\right)^2
\right\}
-
2\frac{d}{dt}
\left(
\frac{ \dot{a} }{a}
\right).
\]

\begin{claim}
\label{Claim-DotG}
The function $g$ with the differential inequality \eqref{Proof-Thm-19-7000} on the interval $[0,T)$ satisfies $\dot{g}\ge0$ on $[0,T)$ 
if $g>0$ and $\dot{g}(0)\ge0$.
\end{claim}

\begin{proof}
Assume that there exists $t\in [0,T)$ such that $\dot{g}(t)<0$.
Put $t_0:=\inf\{t>0;\ \dot{g}(t)<0\}$.
We have $\dot{g}(t_0)=0$ by $\dot{g}(0)\ge0$ and the continuity of $g$.
Thus, we have 
\begin{eqnarray*}
0
&<&
g(t_0)\ddot{g}(t_0)
+(\kappa-2)a(t_0)^{-1}\dot{a}(t_0) g(t_0)\dot{g}(t_0)-\frac{(\kappa+2)(\dot{g}(t_0))^2}{4}, 
\\
&=&g(t_0)\ddot{g}(t_0)
\end{eqnarray*}
by \eqref{Proof-Thm-19-7000},
which shows $\ddot{g}(t_0)>0$ by $g>0$.
So that, we have 
\[
\dot{g}(t_0+\varepsilon)=\int_{t_0}^{t_0+\varepsilon} \ddot{g}(s) ds>0
\]
for sufficiently small $\varepsilon>0$,
which contradicts to the definition of $t_0$.
We have shown $\dot{g}\ge0$ on $[0,T)$.
\end{proof}

We note $\dot{g}(0)>0$ by the assumption \eqref{Thm-19-5000} and \eqref{Proof-Thm-19-4500}.
Since $u_0\not\equiv 0$ by \eqref{Thm-19-5000}, we have $g(0)>0$,
by which $g>0$ on the interval $[0,\varepsilon)$ for small $\varepsilon>0$.
So that, we have $\dot{g}\ge0$ on $[0,\varepsilon)$ by Claim \ref{Claim-DotG}.
Since $t$ with $0<t<T$ and $g(t)=0$ does not exist by \eqref{Proof-Thm-19-7000},
we have $g>0$ on $[0,T)$,
which yields $\dot{g}\ge0$ on $[0,T)$ again by Claim \ref{Claim-DotG}.

We have 
$\ddot{g}>0$ 
by \eqref{Proof-Thm-19-7000}, $\kappa> 2$, $\dot{a}\le 0$, $g>0$ and $\dot{g}\ge 0$.
Thus, we also have 
$\dot{g}(t)>\dot{g}(0)\ge 0$ for $t>0$.

Now, we consider a function $G:=g^{-\kappa_\ast}$ for $0<\kappa_\ast\le (\kappa-2)/4$ with $\kappa>2$.
We have 
\beq
\label{Proof-Thm-19-8000}
\dot{G}=-\kappa_\ast g^{-\kappa_\ast-1}\dot{g}\le 0
\eeq
for $t\ge0$ by $\dot{g}\ge0$.
We also have 
\begin{eqnarray*}
\ddot{G}
&=&
\kappa_\ast
G^{1+2/\kappa_\ast}
\left\{
(\kappa_\ast+1)(\dot{g})^2-g\ddot{g}
\right\}
\\
&<&
\kappa_\ast
G^{1+2/\kappa_\ast}
\left\{
\left(\kappa_\ast-\frac{\kappa-2}{4}\right)(\dot{g})^2
+
\frac{(\kappa-2)\dot{a}\dot{g} g}{a}
\right\}
\end{eqnarray*}
on $[0,T)$ 
by \eqref{Proof-Thm-19-7000}.
By $\dot{g}>0$, $g\ge g(0)$ and $\dot{a}\le 0$, 
we obtain 
\beq
\label{Proof-Thm-19-9000}
\ddot{G}
\le 
-\gamma G(t)^{1+2/\kappa_\ast},
\ \ 
\gamma:=
\kappa_\ast
\left(
\frac{\kappa-2}{4}-\kappa_\ast
\right)
(\dot{g}(0))^2>0
\eeq
on $[0,T)$.
By \eqref{Proof-Thm-19-8000} and \eqref{Proof-Thm-19-9000}, 
we have 
$\dot{G}\le \dot{G}(0)$ and
\[  
G(t)=G(0)+\int_0^t\dot{G}(s) ds\le G(0)+\dot{G}(0)t \to 0
\]
as $t\to T_\ast:=-G(0)/\dot{G}(0)$.
Thus, we have
$g(t)\to\infty$ as $t\to T_\ast$.
So that, we obtain 
$\|u(t)\|_2\to\infty$ as $t\to T_\ast$ by $0<a\le a_0$ on $[0,T_0)$ by $\dot{a}\le 0$.
We note that $T_\ast$ is rewritten as \eqref{Thm-19-6000} by the definition of $G$.
We need the condition $T_\ast\le T_1$ to show the blowing-up occurs in the interval $[0,T_1)$.

\newsection{Proof of Corollary \ref{Cor-20} }
\label{Sec-Cor-20}
Firstly, we check the assumption \eqref{Condition-aM-BlowUp}.
We note that $\dot{a}\le 0$ holds if and only if $H\le0$ by (1) in Lemma \ref{Lem-10}. 
Under $H\le 0$, we note that $\dot{M}\le 0$ holds if and only if $H=0$ with $\sigma\in \br$, or $H<0$ with $-1\le \sigma\le 0$ by (2) in Lemma \ref{Lem-11}.
We note that 
$M^2\ge0$ on $[0,T_\ast)$ holds under one of the following conditions from (i)' to (iv)' by (3), (4) and (5) in Lemma \ref{Lem-11} under $H\le 0$, 
where $T_\ast$ is defined by \eqref{Thm-19-6000}.
\[
\begin{cases} 
(i)' & H=0, \ \sigma\in \br, \  m\ge0, \ T_\ast\le T_0(=\infty). 
\\
(ii)' & H<0, \ \sigma=-1, \ m\ge \frac{n|H|}{2c}, \ T_\ast \le T_0(=\infty).
\\
(iii)' & H<0, \ \sigma=0, \ m\ge0, \ T_\ast\le T_0(<\infty). 
\\
(iv)' & H<0, \ -1<\sigma<0, \ m>\frac{\sqrt{|\sigma|}n|H|}{2c}, \ T_\ast\le T_1(<\infty).
\end{cases}
\]

Secondly, we check the assumption \eqref{Thm-19-3000}.
Put 
\[
I:=
(p-1)\left(
1+\frac{\sigma n^2}{4}
\right)
-n(1+\sigma), 
\ \ 
J:=(p-1)m^2c^2+IH^2
\left(\frac{a}{a_0}\right)^{-n(1+\sigma)}.
\]
The assumption \eqref{Thm-19-3000} is rewritten as $J\ge0$ on the interval $[0,T_\ast)$ 
by Lemma \ref{Lem-10} and (1) in Lemma \ref{Lem-11},
where we have put $\kappa=p+1$.
It suffices to check that this assumption is satisfied under each condition from (i) to (vii) in Corollary \ref{Cor-20}.
Under (i), the condition (i)' holds, and we have $J=(p-1)m^2c^2\ge0$.
Under (ii), the condition (ii)' holds, and we have 
\[
J=(p-1)
\left\{
m^2c^2+
\left(
1-\frac{n^2}{4}
\right)H^2
\right\},
\]
and $J\ge0$ holds if 
$m\ge0$ for $n=1$ and $n=2$, 
or if 
$m\ge \sqrt{n^2/4-1} |H|/c$ for $n\ge3$, 
which are satisfied under (ii)'.

To consider the conditions from (iii) to (vii), let $H<0$ and $-1<\sigma\le0$ under (iii)' and (iv)' in the following.
$J\ge0$ on $[0,T_\ast)$ holds 
if $I\ge0$, 
or if $I<0$, $\two<1$ on $[0,T_\ast)$ with $T_\ast \le T_2(<T_0)$, 
where 
\[
\two:=\sqrt{\frac{-I}{p-1}} \cdot \frac{|H|}{mc}
\]
and we note $T_2$ defined by \eqref{Def-T2} satisfies 
$T_2:=T_0(1-\two)$.
We recall $p_\ast$ and $p_\sharp$ defined by \eqref{Def-pAst-pSharp}. 
The inequality $I\ge0$ holds if and only if 
$\sigma>-4/n^2$ and 
$p\ge p_\ast$ 
by $\sigma>-1$.
The inequality $I<0$ holds if and only if 
$\sigma>-4/n^2$ and $p<p_\ast$, 
or if $\sigma\le -4/n^2$ and $p<\infty$ by $\sigma>-1$ again.
The inequality $\two<1$ holds if and only if 
$I\ge0$ and $p<\infty$,
or if 
$I<0$ and 
\[
\sigma>-\frac{4}{n^2}
\left\{
1+\left(\frac{mc}{H}\right)^2
\right\},
\ \ p_\sharp<p<\infty.
\]
So that, 
$J\ge0$ holds under one of the following conditions when $H<0$ and $-1<\sigma\le0$. 
\[
\begin{cases}
(i)''\ \ \sigma>-\frac{4}{n^2},\ p\ge p_\ast,\ T_\ast\le T_0. \\
(ii)'' \ \ \sigma>-\frac{4}{n^2}, \ p_\sharp<p<p_\ast,\ T_\ast\le T_2. \\
(iii)''\ \ -\frac{4}{n^2}\left\{1+\left(\frac{mc}{H}\right)^2\right\}<\sigma\le -\frac{4}{n^2},
\ p>p_\sharp,\ T_\ast\le T_2.
\end{cases}
\]

Under (iii), the conditions (iii)' and (i)'' hold. 
Under (iv), the conditions (iii)' and (ii)'' hold. 
Under (v), the conditions (iv)' and (i)'' hold. 
Under (vi), the conditions (iv)' and (iii)'' hold. 
Under (vii), the conditions (iv)' and (ii)'' hold. 
So that, the assumptions \eqref{Condition-aM-BlowUp} and \eqref{Thm-19-3000} hold under the conditions from (i) to (vii) in Corollary \ref{Cor-20}, and the required result follows from Theorem \ref{Thm-19}.

%

\end{document}